\pdfoutput=1

\newif\ifabstract
\abstracttrue
\abstractfalse  %
\newif\iffull
\ifabstract \fullfalse \else \fulltrue \fi

\documentclass[11pt]{article}
\usepackage{hyperref}
\usepackage{fullpage,graphicx,latexsym,url}

\newcommand\keywords[1]{\paragraph{Keywords:} #1\par}
\newcommand\numberofauthors[1]{}
\newcommand\affaddr[1]{#1}
\newcommand\alignauthor{}
\let\email=\url

\newenvironment{proof}{\noindent\textbf{Proof: }\ignorespaces}
  {\hspace*{\fill}$\Box$\medskip}

\usepackage{comment}
\usepackage{amsmath}
\usepackage{amssymb}
\usepackage{subfig}

\newtheorem{theorem}{Theorem}[section]
\newtheorem{lemma}[theorem]{Lemma}

\newtheorem{conjecture}[theorem]{Conjecture}
\newtheorem{corollary}[theorem]{Corollary}

\begin{document}

\title{One Tile to Rule Them All: \\
       Simulating Any Turing Machine, Tile Assembly System, \\ or Tiling System with a Single Puzzle Piece}

\author{
\numberofauthors{7}
\alignauthor
  Erik D. Demaine\thanks{Research supported in part by NSF grant CDI-0941538.}\\
    \affaddr{CSAIL}\\
    \affaddr{MIT}\\
    \affaddr{Cambridge, MA 02139, USA}\\
    \email{edemaine@mit.edu}
\and
\alignauthor
  Martin L. Demaine\footnotemark[1]\\
    \affaddr{CSAIL}\\
    \affaddr{MIT}\\
    \affaddr{Cambridge, MA 02139, USA}\\
    \email{mdemaine@mit.edu}
\and
\alignauthor
  S\'andor P. Fekete\\
    \affaddr{Computer Science}\\
    \affaddr{TU Braunschweig}\\ %
    \affaddr{Germany}\\
    \email{s.fekete@tu-bs.de}
\and
\alignauthor
  Matthew J. Patitz\thanks{Research supported in part by NSF grant CCF-1117672.}\\
    \affaddr{Dept. of CSCE}\\
    \affaddr{University of Arkansas}\\
    \affaddr{Fayetteville, AR 72701, USA}\\
    \email{patitz@uark.edu}
\and
\alignauthor
  Robert T. Schweller\footnotemark[2]\\
    \affaddr{Computer Science}\\
    \affaddr{University of Texas--Pan American}\\
    \affaddr{Edinburg, TX 78539, USA}\\
    \email{rtschweller@utpa.edu}
\and
\alignauthor
  Andrew Winslow\thanks{Research supported in part by NSF grant CDI-0941538.}\\
    \affaddr{Computer Science}\\
    \affaddr{Tufts University}\\
    \affaddr{Medford, MA 02155, USA.}\\
    \email{awinslow@cs.tufts.edu}
\and
\alignauthor
  Damien Woods\thanks{Supported by NSF grants 0832824 (the Molecular Programming Project), CCF-1219274, and CCF-1162589.}\\
    \affaddr{Computer Science}\\  %
    \affaddr{California Institute of Technology}\\
    \affaddr{Pasadena, CA 91125, USA}\\
    \email{woods@caltech.edu}
}
\date{}
\maketitle


\date{}

\begin{abstract}
In this paper we explore the power of tile self-assembly models that extend the
well-studied abstract Tile Assembly Model (aTAM) by permitting tiles of shapes
beyond unit squares.  Our main result shows the surprising fact that any aTAM
system, consisting of many different tile types, can be simulated by a
\emph{single} tile type of a general shape.
As a consequence, we obtain a single \emph{universal} tile type of a single
(constant-size) shape that serves as a ``universal tile machine'': the single
universal tile type can simulate any desired aTAM system when given a single
seed assembly that encodes the desired aTAM system.
We also show how to adapt this result to convert any of a variety of plane
tiling systems (such as Wang tiles) into a ``nearly'' plane tiling system
with a single tile (but with small gaps between the tiles).
All of these results rely on the ability to both rotate and translate tiles;
by contrast, we show that a single nonrotatable tile, of arbitrary shape,
can produce assemblies which either grow infinitely or cannot grow at all, implying drastically
limited computational power.
On the positive side, we show how to simulate arbitrary cellular automata
for a limited number of steps using a single nonrotatable tile
and a linear-size seed assembly.
\end{abstract}

%
%
%

\pagenumbering{arabic}
\pagestyle{plain}

\section{Introduction}

This paper shows that %
a single rotatable tile type suffices
to simulate any desired tile system (consisting of many different tile types)
in two different worlds:
classic geometric tiling and self-assembling tile systems.

\paragraph{Classic geometric tiling}
Tiling the plane with geometric shapes goes back to Kepler in 1619;
refer to the classic survey book \cite{Gruenbaum-Shephard-1987}.
Wang \cite{Wang-1961} introduced the geometrically simpler model of tiling the
plane with (fixed-orientation) unit-square tiles with colored edges
such that abutting edges match in color.  Such edge-matching tiles can
be converted to purely geometric tiles,
using jigsaw-like tabs and pockets to simulate colors;
a reverse translation is possible by losing some factors in scale
and number of tiles~\cite{Demaine-Demaine-2007-jigsaw}.

The tiling literature often aims for minimal tile sets to achieve certain
properties.  For example, Wang's student Berger \cite{Berger-1964} first found an
aperiodic set of 20,426 Wang tiles, and a sequence of improvements by
Berger, Knuth, L\"auchli, Robinson, Penrose, Ammann, and Culik finally
resulted in an aperiodic set of just 13 Wang tiles \cite{Culik-1996}.
Using geometric tiles (or Wang tiles that can rotate and reflect, and a complementary matching condition), Robinson~\cite{Robinson-1971} obtained just 6 tiles,
while the famous Penrose tilings \cite{Penrose-1979} achieve aperiodicity
with just two tiles.  It remains open whether one tile suffices
to obtain aperiodicity; such a result has been obtained given
nonlocal matching rules among tiles \cite{Socolar-Taylor-2012}, or with overlaps \cite{gummelt1996penrose}.

Tilings are closely related to computation.
Berger \cite{Berger-1966} proved that it is undecidable to determine whether a
set of Wang tiles (with infinitely many copies of each) tile the plane, by
simulating a Turing machine; Robinson \cite{Robinson-1971} gave a simpler proof.
Both proofs use an unbounded number of Wang tiles, growing with
the size of the Turing machine.

\paragraph{Our results I: Universal tiling with one tile}
We prove in Section \ref{Plane Tilings}
that these and many other plane tiling systems, consisting of any
(finite) number of tiles, can be simulated by a ``nearly plane'' tiling that uses only
a single tile.  The input tiling system can live on either a square or
hexagonal grid, can allow tiles to rotate or not, can allow tiles to reflect
or not, and can define compatible tile adjacencies according to matching
color (as in Wang tiling) or complementary color pairs (as in actual DNA).
The output one-tile system requires tiles to live on the same square or
hexagonal lattice, allows tiles to rotate, and is \emph{nearly plane} tiling
in the sense that it leaves tiny gaps between the tiles.

Our simulation result applies in particular to all Wang tilings, even allowing
rotation as in Robinson's tilings, which has several important implications:
\begin{enumerate}
\item We obtain for the first time (in any model of tiling) that a single tile
can simulate any desired Turing computation.
(Here the combinatorial complexity of the tile, instead of the number of tiles,
depends on the size of the Turing machine.)
\item We obtain for the first time (in any model of tiling with a local
matching rule and no overlaps) that a single tile suffices to produce aperiodic
(nearly plane) tilings.
\item We obtain for the first time (in any model of tiling) that a single tile
can simulate a countably infinite number of tiling systems (a form of
intrinsic universality).
\end{enumerate}

Before proving these results, we develop a number of techniques for another model:
tile self-assembly. In fact, the above results are essentially achieved as
corollaries, with some additional technical detail, of our tile self-assembly
results, the proofs of the latter being somewhat more involved and constitute the main portion of the paper.

\paragraph{Tile self-assembly}
Winfree \cite{Winf98} introduced the \emph{abstract Tile Assembly Model}
(aTAM) as a clean theoretical model
\ifabstract
\else
that %
approximates the reality of
\fi
nanoscale self-assem\-bling systems.  In several experiments of
increasing complexity and reliability (e.g., \cite{SchWin04,SchWin07}),
this model has been shown to be physically practical,
with tiles made up from DNA strands.
As a result, the model has become the standard
in theoretical work on self-assembly. %

The aTAM is essentially a more local (and thus more realistic)
version of Wang tiling:
we start with a specific tile (called the \emph{seed}), and repeatedly add
any tile to the assembly that has enough matching glues (colored edges) to
``stick'' to the rest of the assembly.  Here we assign an integer
\emph{strength} to each glue type (color class), which represents the
affinity/attraction for matching glues of that type, and specify a global
\emph{temperature} (typically~$2$) specifying the total required strength
for a tile to attach to the assembly. Unlike Wang tiling, in the aTAM we can never throw away partially formed assemblies and, in fact, the aTAM can be thought of as a special kind of asynchronous, and nondeterministic, cellular automaton. See Section~\ref{Models} for details.

\paragraph{Our results II: Universal self-assembly with one tile}
We prove in Section~\ref{sec:manygons} that any aTAM system---consist\-ing of
any (finite) number of tiles---can be simulated by just a single tile,
in a generalization of the aTAM model called pfbTAM
(polygonal free-body Tile Assembly Model).
Precisely, we show that any temperature-$\tau$ aTAM system can be converted
into a temperature-$\tau$ pfbTAM system of a single tile type such that the two systems have exactly
the same producible assemblies (modulo isometry).
This construction is \emph{self-seeding} in the sense that it starts from a
single copy of the very same tile; it is even a challenge to get the next copy
of the tile to attach without allowing the construction to grow indefinitely.

Combining this result with intrinsic universality of aTAM \cite{IUSA},
we construct
a single constant-size tile $t_U$ such that
any temperature-$\tau$ aTAM system $\Gamma$ can be converted into a
seed assembly such that adding (infinitely many copies of) tile $t_U$
to the seed in the temperature-2 pfbTAM simulates~$\Gamma$.
This tile $t_U$ is a kind of geometric analog to a universal Turing machine,
simultaneously simulating the shape construction and computational ability
of an arbitrary tile assembly system.
(On the other hand, no standard universal Turing machine has only one state, or only one symbol~\cite{neary2012complexity}.)

The pfbTAM model differs from aTAM in two ways---allowing polygonal tiles
instead of just squares and allowing tiles to rotate---both of which
we claim are physically realistic.
For example, DNA origami \cite{RotOrigami05} is
a rapidly evolving technology that has been used to  successfully build numerous
complex shapes using strands of DNA.  The technology has evolved to the point
where free software automatically designs DNA to fold into essentially
arbitrary desired shapes.  aTAM-like tiles with polyomino
(instead of square) shapes have already been developed in practice
\cite{WooRoth11} and studied in theory \cite{fu2012SAGT}.
Second, rotation is a natural addition, as nanoscale objects do rotate in
reality; the aTAM omitted rotation, as it was easier to work with translation
only, and with only square tiles, it turns out to be equivalent.

Given that our tile rotates and translates, a natural question is what can be achieved with {\em translation only}?  We show that allowing rotation is essential:
the same type of results are impossible with translation-only polygonal tiles.
Specifically, we prove in Section~\ref{sec:inmemoryofporky} that any
single-tile translation-only pfbTAM self-assembly system either can grow
infinitely (producing assemblies of unbounded size) or cannot grow at all
(producing just the seed tile).  This negative result holds even if the system
starts from a more complex seed assembly consisting of up to three copies of
the tile, even though such a seed avoids the difficult issue of getting the
first bond between two tiles.
On the positive side, though, in Section~\ref{sec:translation-positive} we show that a seed assembly consisting of
$O(k)$ copies of the single tile suffices to simulate an arbitrary 1D cellular
automaton for $k$ steps.  In general, we conjecture that no finite seed
suffices for unbounded computational power with translation only,
in stark contrast to allowing rotation.

\paragraph{\hspace{-0.28ex}Our results III:~Hexagonal tile assembly systems}

Along the way, we consider in Section~\ref{sec:hex}
aTAM systems with unit-sized hexagonal tiles on a hexagonal grid.
The only previous paper considering this model
\cite{KSX12} simply showed differences between squares and hexagons
with respect to infinite constructions.
Here we show that any temperature-$2$ square aTAM system can be simulated by a
temperature-$2$ hexagonal aTAM system in which all glues have strength
at most~$1$.
The construction works at a scale factor of only~3: each square tile is simulated by a $3 \times 3$ block of hexagonal tiles.
By contrast, any temperature-$2$ square aTAM system in which all glues have
strength at most $1$ can not grow outside its bounding box;
we use the fact that hexagons have the ability to continue to grow.

This result is a key step to proving our main positive result
(aTAM simulation allowing translation and rotation).
Specifically, we show in Section~\ref{sec:manygons}
how to simulate any tem\-pera\-ture-$2$ hexagonal aTAM system, that uses
strength-$1$ glues, by a rotatable polygon that mimics different tile types by
attaching at different rotations.
More generally, we believe that our result on hexagonal aTAM
may have independent importance in the study of self-assembly systems.

\section{Models}
\label{Models}
\label{sec:models}

\subsection{The Polygonal Free-Body Tile Assembly Model}

The \emph{polygonal free-body Tile Assembly Model} (or \emph{pfbTAM}) generalizes self-assembly models
such as the aTAM by using
arbitrary geometric shapes for the constituent parts that can be translated and rotated by any amount.
In our positive results, we focus on convex regular $n$-gons with small surface geometries; our negative results are valid for arbitrary shapes, as discussed in Section~\ref{sec:inmemoryofporky}.

The basic units of the pfbTAM are  polygonal tiles, and a pfbTAM system $\Gamma$ is defined as $\Gamma = (F, \tau, \sigma)$, where $F$ is a finite set of polygonal tile types, $\tau \in \mathbb{N}$ is the {\em temperature} parameter for the system, and $\sigma$ is a {\em seed assembly} consisting of a set of polygonal tiles from $F$ and their locations.

\paragraph{Tiles}

A {\em polygonal tile} in the pfbTAM model
is bounded by a simple closed polygon enclosing the tile's interior.
Boundaries of two tiles may intersect, but not their interiors.

\paragraph{Glues}

Let $\Sigma$ be an alphabet of \emph{glue types}.
Each glue type $g \in \Sigma$ is assigned a value \emph{strength} ($g\in \mathbb{N}$).
The boundary of the polygon tile is divided into intervals called \emph{sides} (which may be more complex than single line segments), and each interval is assigned a glue type from $\Sigma$.
A pair of boundary intervals are \emph{coincident} if there is a bijection between pairs of points on the two intervals such that each pair is coincident.
For a given pair of sides on distinct polygon tiles, we define three types of geometric compatibility, based upon whether the tiles can be oriented such that portions of their boundaries are coincident.
See Figure~\ref{fig:manygon-geometry-pair-defns} for a description of geometric compatibility.

\begin{figure}[htp]
    \begin{center}
    \includegraphics[width=.9\columnwidth]{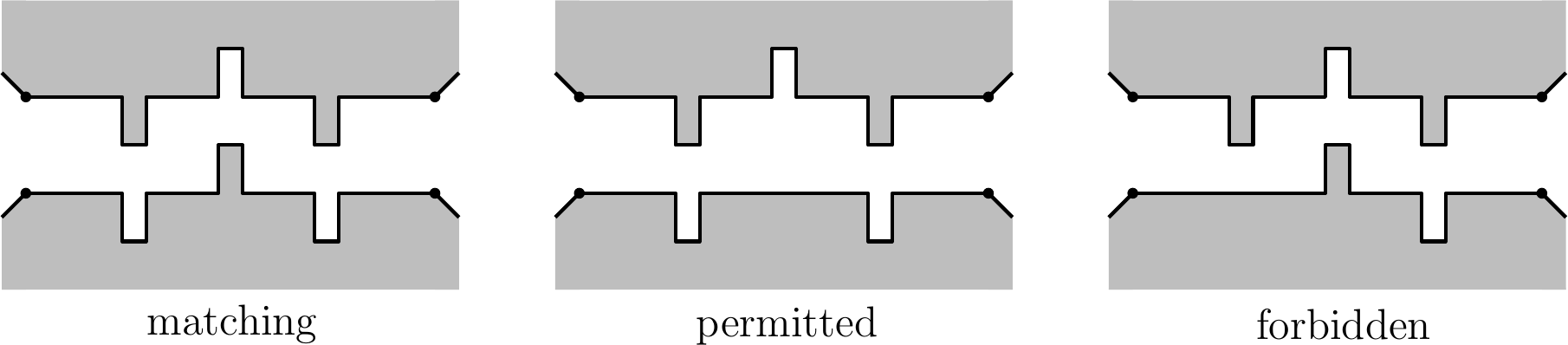}
    \caption{
    \label{fig:manygon-geometry-pair-defns}
    The three possible types of side-side interactions between a pair of polygon free-body tiles.
    Matching sides form a bond with strength according to their glues, while permitted and forbidden sides do not form a bond.}
    \end{center}
\ifabstract
\else
 \vspace*{-2mm}
\fi
\end{figure}

If there exists a pair of tile orientations so that the pair of sides become fully coincident, then these sides have \emph{matching geometry}.
If a pair of sides do not have matching geometry, but can be oriented such that the endpoints (and possibly more) of the sides are coincident, then the sides have \emph{permitted geometry}.
Finally, if a pair of sides do not have permitted geometry, then they have \emph{forbidden geometry}.

Bonding between tiles occurs via the interaction of side pairs.
If a pair of sides have permitted or forbidden geometry, then they cannot form a bond, regardless of glues.
If the sides instead have matching geometry, then they can bond with strength according to their glues.
Matching glues form a bond with strength determined by the glue type.
Non-matching glues form zero-strength bonds.

\paragraph{Systems and assembly}

An {\em assembly} is a collection of tiles that are bound together by their
adjacent glues; an assembly is {\em $\tau$-stable} if any possible way of
separating it into two disconnected assemblies must break glue bonds
that sum to at least the value $\tau$.  Formally, we define the \emph{bond
graph} of an assembly to be the planar graph consisting of a labeled node, and
an edge between two tiles iff they have edges that form a bond.
A tile can attach to an assembly if it is able to bind with strength at least $\tau$.

Assembly of a pfbTAM system $\Gamma = (F, \tau, \sigma)$ starts with the seed
assembly $\sigma$ and proceeds as individual tiles from $F$
nondeterministically bind $\tau$-stably and one at a time to the growing
assembly. An assembly is \emph{producible} if there is a sequence of
$\tau$-stable assemblies that produces it; a producible assembly is {\em
terminal} if there are no tiles that can $\tau$-stably attach to it.
The set of producible assemblies of $\Gamma$ is $PROD_\Gamma$, and the set of
terminal assemblies is $TERM_\Gamma$.  A $\Gamma$ is {\em directed}
(a.k.a. \emph{deterministic}, \emph{confluent}) if there is a single
producible, terminal assembly (i.e. $|TERM_\Gamma| = 1)$.

\paragraph{Rotationally restricted systems}  In some cases we restrict the granularity in which polygonal tiles may rotate, and whether or not they can flip.  A $\phi$-restricted system limits $PROD_\Gamma$ to only contains assemblies that are obtained through attachments of tiles in $F$ that are translations and flips and rotations of degrees $c \cdot \phi$ for integers $c$.  Rotationally restricted systems with $\phi = 0$ that do not allow flipping are referred to as \emph{translation-only} systems.

\subsection{aTAM and hTAM}

\paragraph{aTAM} The abstract Tile Assembly Model can be formulated as a restricted version of the pfbTAM.  An aTAM system is a translation-only pfbTAM system $\Gamma = (T,\tau,s)$ such that $T$ consists of equal size squares in the same rotational orientation with one glue type assigned to each tile side.

\paragraph{hTAM} The hexagonal Tile Assembly Model is a restricted version of the pfbTAM.  An hTAM system is a trans\-lation-only pfbTAM system $\Gamma = (T,\tau,s)$ such that $T$ consists of equal size regular hexagons in the same rotational orientation with one glue type assigned to each tile side.

\subsection{ Simulation }
In this section we define what it means for the pfbTAM to simulate the hTAM and aTAM and for the hTAM to simulate the aTAM.

\paragraph{Simulating hTAM systems with pfbTAM systems}

At a high level, we say a pfbTAM system $\Gamma_f = (T', \tau', \sigma)$ simulates an hTAM system $\Gamma_h = (T, \tau, \sigma)$ if there is a bijection between orientations (composed of rotations and flips) of tiles in $T'$ and tiles in $T$ such that this bijection yields a second bijection between producible assemblies of $\Gamma_f$ and $\Gamma_h$, where both tile location and bond structure is preserved.

A pfbTAM system $\Gamma_f = (T', \tau', \sigma')$ \emph{simulates} hTAM
system $\Gamma_h = (T, \tau, \sigma)$ if there's a mapping $\phi: T' \times
[0, 2\pi) \times \{R, R'\} \rightarrow T$ of orientations of tiles (specified
by an angle in $[0, 2\pi)$ and one of two reflections $R$ or $R'$) in
$T'$ to tiles in $T$ such that for every bond graph $G'$ generated by an
assembly produced by $\Gamma_f$, there's a bond graph $G$ produced by an
assembly of $\Gamma_h$, iff $G'$ yields $G$ when each node is mapped
via $\phi$. Also, for each assembly $A_f'$ produced by $\Gamma_f$ as a single
tile addition from assembly $A_f$, an assembly $A_h'$ in $\Gamma_h$ equivalent to
$A_f'$ via $\phi$ can be produced from $A_h$ (equivalent to $A_f$ via $\phi$) by a single tile addition, and vice versa.
\ifabstract
\else
(Namely, the exact same assemblies are producible in both systems, in exactly the same orderings
of tile additions.)
\fi

\paragraph{Simulating aTAM systems with hTAM systems}

An assembly $A_h$ over hexagon tiles $T_h$ is a {\em valid $c$-block representation} for odd, positive integer $c$ and partial function $\phi: T_h \rightarrow T_a$  if (1) $A_h$ is evenly divisible into $c \times c$ blocks of tiles as shown in Figure~\ref{fig:hex1a-layout}, and (2) $\phi(x)$ is defined for hex tile $x$ in $A_h$ iff $x$ is at the center of a $c\times c$ block.

For a $c$-block representation $A$, define the $c$-bond graph of $A$ to consist of a vertex for each center hex tile $x$ of each $c\times c$ block with node label $\phi(x)$.  Two hex tile vertices are connected by an edge if there exists a connected path of bonded hex tiles connecting the two vertices of length exactly $c$ that consists of purely north jumps, or south jumps, or southeast jumps, or northwest jumps.

An hTAM system $\Gamma_h = (T_h, \tau_h , \sigma_h)$ simulates aTAM system
$\Gamma_a = (T_a , \tau_a , \sigma_a)$ at scale $c$ for positive, odd integer
$c$, if there's a partial function $\phi: T_h \rightarrow T_a$, such that (1)
every tile in any producible assembly of $\Gamma_h$ which is of size greater than $c^2 - 1$ is
within distance $d \leq c$ from a tile $x$ for which
$\phi(x)$ is defined, and (2)  there's a producible assembly $A_h$ of
$\Gamma_h$ that's a valid $c$-block representation for function $\phi(x)$ with
$c$-block bond graph $G_h$ iff there's a producible assembly $A_a$ of
$\Gamma_a$ with the same bond graph $G_h$.  Further, for each producible assembly $A_a$  of $\Gamma_a$ which
grows into some $A_a'$ via a single tile addition, there are equivalent $c$-block representation assemblies $A_h$ and $A_h'$ of $\Gamma_h$
such that $A_h$ grows directly into $A_h'$ via some number of tile additions.  Note: during several tile additions, namely
those which don't fill positions where $\phi(x)$ is defined, the assembly of $\Gamma_h$ will still
map to $A_a$.  Vice versa, for any pair of assemblies $A_h$ and $A_h'$ of $\Gamma_h$ such that $A_h$ grows
into $A_h'$ via a single tile addition and the $c$-block representations of $A_h$ and $A_h'$ map to different assemblies in $\Gamma_a$, there exist assemblies $A_a$ and $A_a'$ of $\Gamma_a$ such that $A_a$ grows into $A_a'$ via a single tile addition.

\paragraph{Simulating aTAM systems with pfbTAM systems}
We define a $c$-scaled simulation of an aTAM system by a pfbTAM system by mapping $c\times c$ blocks within pfbTAM assemblies to aTAM tiles, where this mapping reads rotations and reflections of pfbTAM tiles in the blocks. The following definition is based on the more formal definition of~\cite{IUSA}.

A pfbTAM system $\Gamma_f = (T', \tau', \sigma')$ \emph{simulates} an  aTAM
system $\Gamma_a = (T, \tau, \sigma)$ at scale $c\in \mathbb{N}$, if both systems have equivalent {\em production} and {\em dynamics} under a representation function $\phi$ defined as follows.

(1) Production:  there is a mapping $\phi: {({T'\cup \{\mathrm{empty}\}} \times [0, 2\pi) \times \{R, R'\}})^{c^{2}} \rightarrow T \cup \{\mathrm{empty}\}$  of $c \times c$ blocks of tiles from~$T'$ and possibly empty locations (where $\phi$ is defined on the orientations of those tiles, specified by a rotation angle in $[0, 2\pi)$ and one of two reflections~$R$ or~$R'$) to tiles in $T$ (or empty locations) such that for every producible assembly~$\pi$ in~$\Gamma_f$ there is a producible assembly $\alpha$ in $\Gamma_a$ where $\alpha =  \phi^{\ast}(\pi )$, and for every  producible assembly $\alpha$ in $\Gamma_a$ there exists a producible assembly $\pi$ in $\Gamma_f$ where  $\alpha = \phi^{\ast}(\pi )$ (here~$\phi^{\ast}$ denotes the function~$\phi$ applied to an entire assembly, in the most obvious block-wise way).  We also require that~$\pi$ maps {\em cleanly} to $\alpha$ under $\phi^{\ast}$, that is, for all non-empty $c \times c$ blocks $b$ in $\pi$ it is the case that at least one neighbor of $\phi(b)$ in $\phi^{\ast}(\pi)$ is non-empty, or else~$\pi$ has at most one non-empty $c \times c $ block. In other words, $\pi$ may have tiles in $c \times c$ blocks representing empty space in $\alpha$, but only if that position is adjacent to a tile in $\alpha$.

(2) Dynamics: if there exist producible assemblies $\alpha$ and $\alpha'$ in $\Gamma_a$ such that $\alpha \rightarrow_1 \alpha'$ (growth by single tile addition), then for every producible $\pi$ in $\Gamma_f$ where $\alpha = \phi^{\ast}(\pi )$ it is the case that there exists~$\pi'$  such that $\pi \rightarrow_{\ast} \pi'$ (growth by one or more tile additions) in $\Gamma_f$ where $\alpha' = \phi^{\ast}(\pi')$. Furthermore, for every pair of producible assemblies $\pi$, $\pi'$ in $\Gamma_f$, if $\pi \rightarrow_{\ast} \pi'$, and $\alpha = \phi^{\ast}(\pi)$ and $\alpha' = \phi^{\ast}(\pi')$, then $\alpha \rightarrow_{\ast} \alpha'$ for assemblies $\alpha, \alpha'$ in $\Gamma_a$.

\subsection{Plane tiling}\label{def:planetiling}

Plane tiling systems, such as Wang tiles~\cite{Wang-1961} and Robinson tiles~\cite{Robinson-1971} consist of sets of shapes (called \emph{tiles}), placed on a regular lattice, so that they cover the entire plane in an infinite arrangement.
A \emph{plane tiling system} $(S, L, T, C)$  is a set of tiles $S$, specified both with shape and (optionally) color patterns on the boundary of each shape, a square or hexagonal lattice $L$, a set of transformations $T$ that the tiles can undergo that necessarily includes translation (the tiles must translate to locations throughout the plane) and optionally rotation and reflection, and a set of tile adjacency constraints $C$ that requires that color patterns on adjacent tiles must either match or be paired in complementary pairs.

A \emph{plane tiling family} $(T, C)$ consists of all plane tiling systems where the tiles are permitted to undergo a set of transformations and coincident tile edges must obey a set of tile adjacency constraints.
The set of transformations include translation and a set $T \subseteq \{t_r, t_f\}$ of optional transformations: rotation and reflection, denoted $t_r$ and $t_f$, respectively.
The tile adjacent constraint $C \in \{ c_c, c_m \}$ is either that adjacent edges of tiles must match ($C = c_m$) or be paired with their complementary geometry/color pattern ($C = c_c$).
We also define a \emph{nearly-plane tiling system} to be a tiling system with the relaxed constraint that tiles need only be placed at every lattice location and touch neighbors specified by the lattice graphs, but need not fill the plane.

\section{Low-Strength Hexagons Simulate High-Strength Squares}\label{sec:hex}
\begin{figure*}[htp]
\centering
\ifabstract
\else
\newcommand{\dwfigspace}{12} %
\fi
  \subfloat[][Square tiles (left) simulated by $3\times 3$ hexagonal supertiles (right).~4 strength $< \tau$ input glues are indicated by  black rectangles. Potential input hexagon sides to the  empty white $3\times 3$ supertile are highlighted in bold.~To simulate a strength $< \tau$ glue, each superside places a single hexagon tile, indicated by a small colored rectangle, with the  goal of cooperatively claiming the center (gray) location.]{%
        \label{fig:hex1a-layout}%
        \centering
       \hspace{\dwfigspace pt}\hspace{6pt}\includegraphics[height=1.5in]{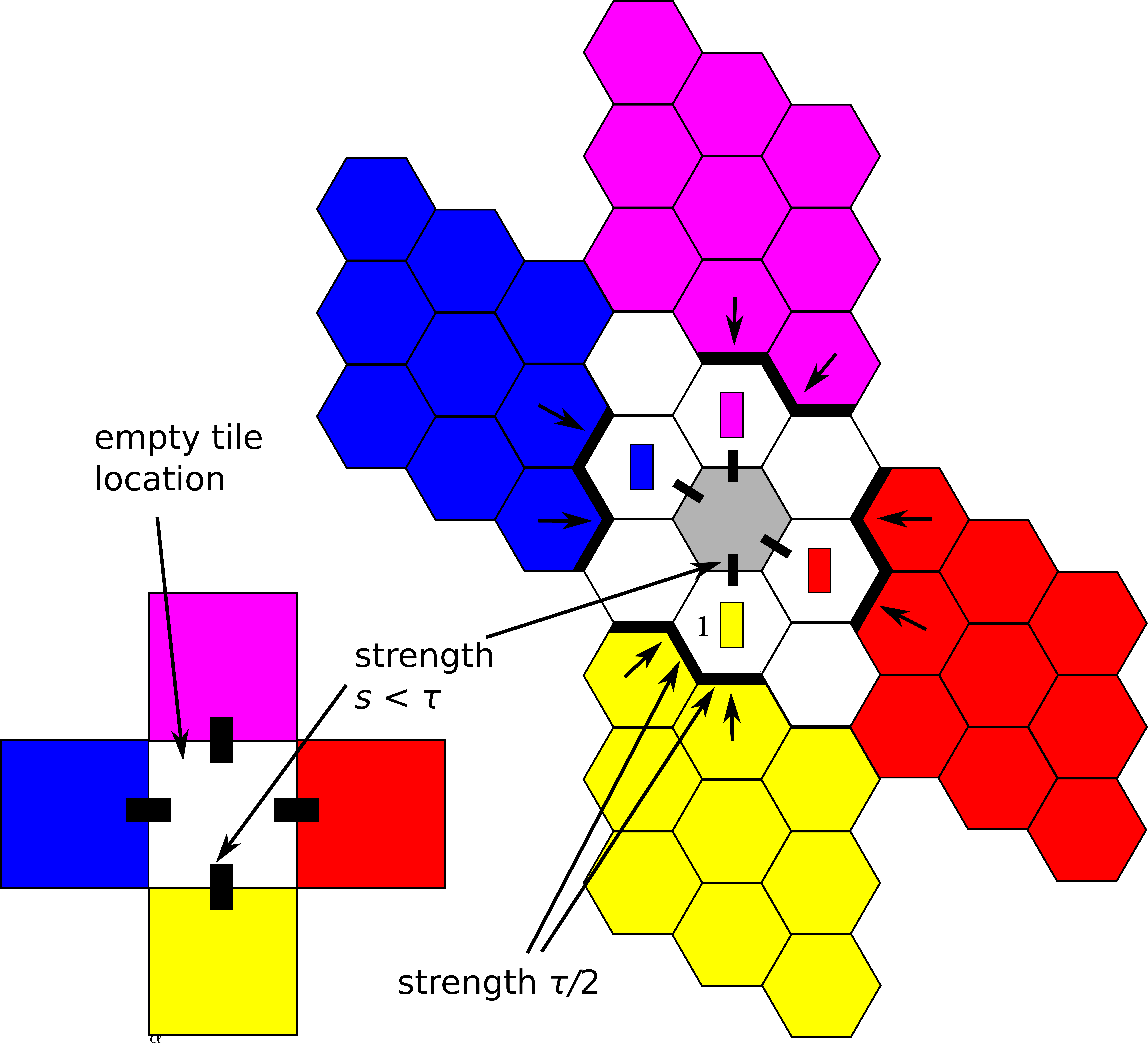}\hspace{\dwfigspace pt}\hspace{6pt}}%
        \hspace{5pt}
  \subfloat[][Potential output hexagon sides from the white $3\times 3$ supertile are highlighted with bold gray lines. A valid simulation must place output glues at these sides, unless the side already contains an adjacent (input) side that came from the gray supertiles.]{%
        \label{fig:hex1-outputs}%
        \centering
       \hspace{\dwfigspace pt}\includegraphics[height=1.5in]{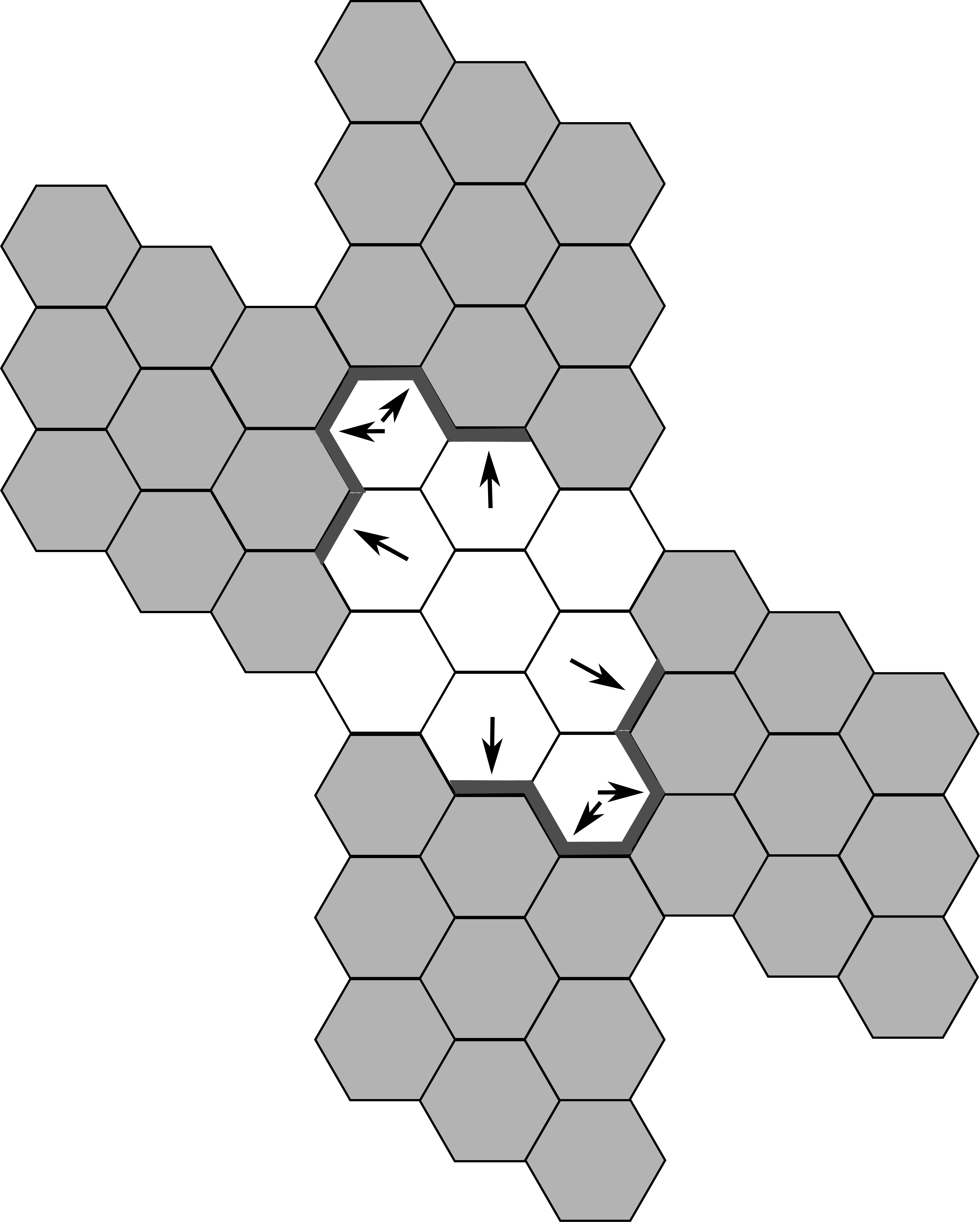}\hspace{\dwfigspace pt}}%
        \hspace{5pt}
  \subfloat[][North \& South strength $< \tau$ cooperating supersides. After the center tile is cooperatively claimed by North \& South, the remaining 4 tiles are placed,  encoding  output supersides.]{%
       \label{fig:hex--NS-less-than-tau}%
        \hspace{\dwfigspace pt} \includegraphics[height=1.5in]{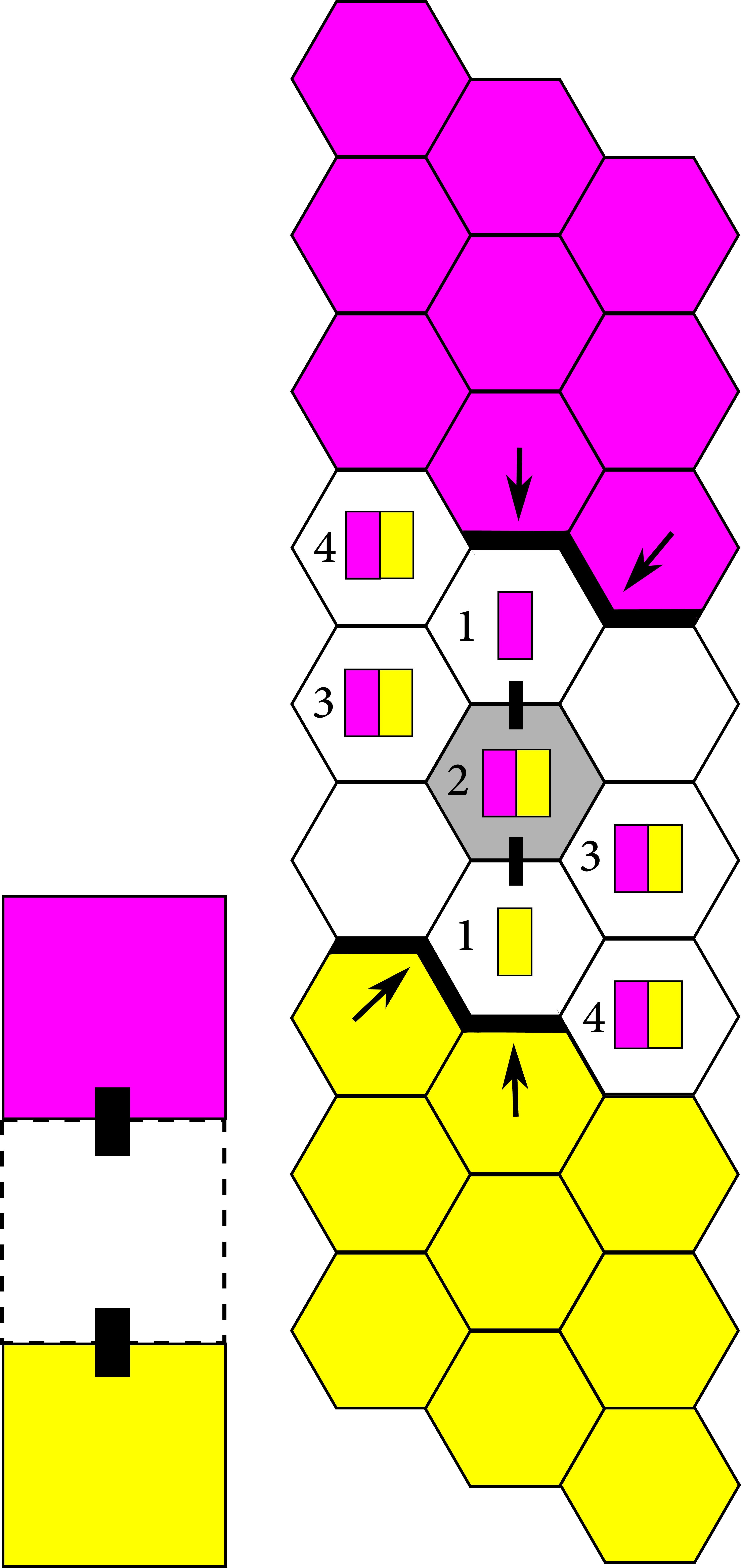}\hspace{\dwfigspace pt}}%
  \hspace{5pt}
  \subfloat[][West \& South strength $< \tau$ cooperating supersides. After the center tile is cooperatively claimed by North \& South, the remaining 4 tiles are placed, encoding suitable output supersides.]{%
        \label{fig:hex--NW-less-than-tau}%
        \centering
       \hspace{\dwfigspace pt}\includegraphics[height=1.29in]{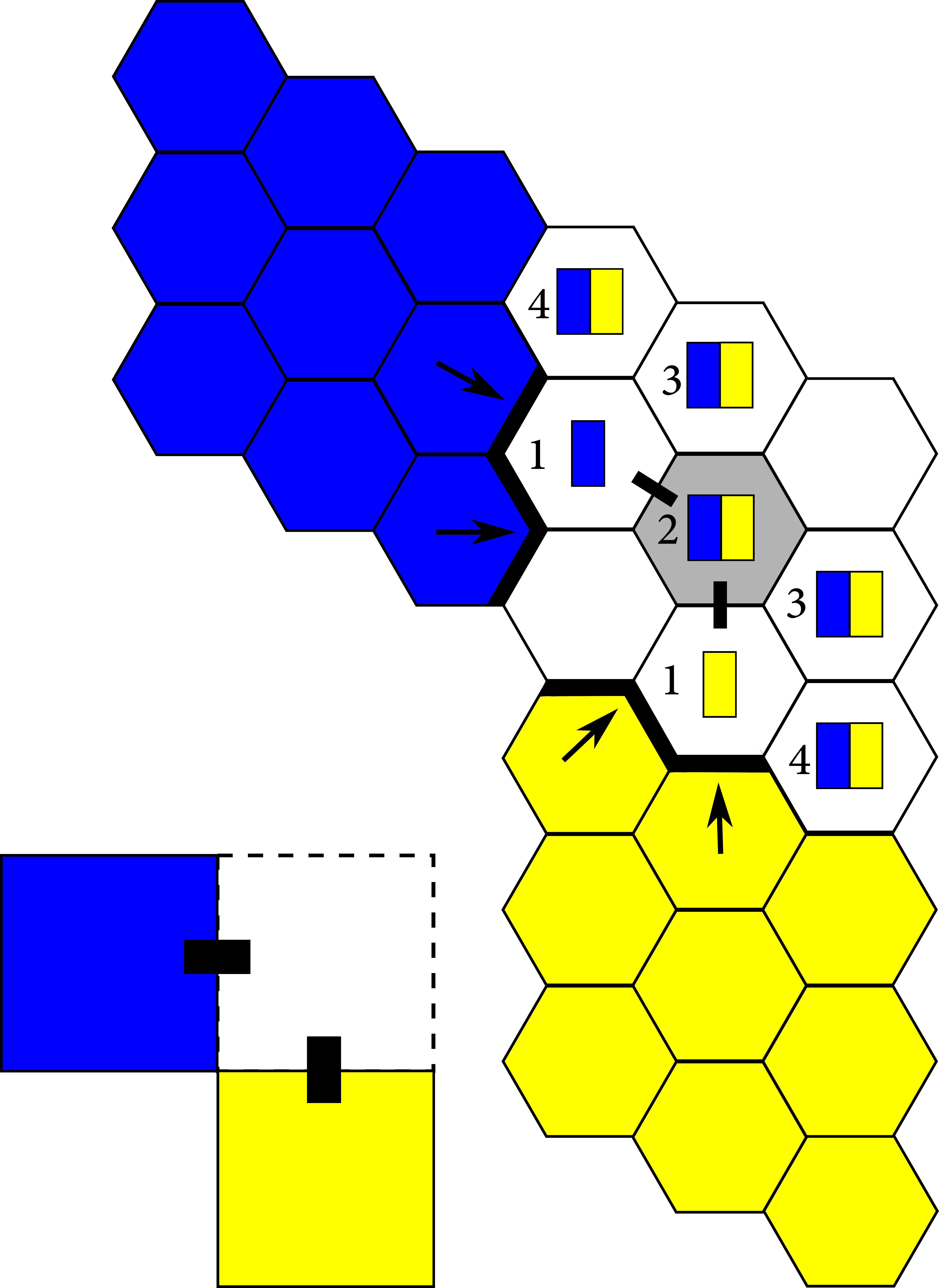}\hspace{\dwfigspace pt} } \vspace{-1ex}%
      \caption{Simulating  strength $< \tau$  squares with a temperature $\tau$ hexagon system with strength $< \tau$ glues. All hexagon glues are of strength $\lceil t/2 \rceil$, except glues bordering the center tile whose strengths are those of the simulated square input glues.}
  \label{fig:hex1}
  \centering
\end{figure*}

The following lemma states that any temperature $\tau$ aTAM system can be simulated by a temperature $\tau$ hTAM system that uses only glues of strength $<\tau$.
\ifabstract
\else
In this section we present the details of the simulation and an argument for its correctness.
\fi

\begin{lemma}
For any aTAM system $\Gamma = (T, \tau, \sigma)$ with $|\sigma| = 1$ and $\tau > 1$, there exists a hexagon assembly system $\Gamma' = (H, \tau, \sigma')$  that simulates $\Gamma$ and has the property that all glues  in the hexagonal-tile set $H$ are of strength  $< \tau $. Also, $|H| = O(|T|^2)$,  $|\sigma'| = 3$ and the simulation has a constant scale blowup factor of $3 \times 3$.
\label{lem:hex}
\end{lemma}

\begin{proof} (Lemma~\ref{lem:hex}) \paragraph{Representation}
Let $\Gamma = (T,\tau,\sigma)$ be any aTAM system.  We create an hTAM
system $\Gamma' = (H,\tau,\sigma')$ that simulates $\Gamma$ at a scale factor
of $3$, i.e., each tile from $T$ in an assembly of $\Gamma$ is
represented by a $3 \times 3$ square of $9$ hexagonal tiles from $H$ in an
assembly of $\Gamma'$.  Each such $3 \times 3$ block in
$\Gamma'$ is a {\em supertile}.  The hexagonal plane is logically divided into
supertile blocks (e.g., see Figure~\ref{fig:hex1} for 5 supertile hex
blocks simulating 5 square tile locations) so that each supertile has $4$
supertile neighbors (N, E, S, and W), which is partially accomplished by
``ignoring'' the northeast edge of the northeast tile and the southwest edge of
the southwest tile of each supertile. For each $t \in T$, there are several
supertiles in $\Gamma'$ that represent $t$; there must be a way to map an
assembly in $\Gamma'$ to an assembly in $\Gamma$.  This is done with a mapping
function that maps each supertile to some $t$ by identifying its center tile.
For each $t \in T$ there will potentially be several tile types $h \in H$ which
can be placed as a center tile of a supertile and map the entire supertile to
that $t$, but for each such $h$, there is exactly one $t$ to which it maps.
Furthermore, if a supertile is incomplete and does not have a center tile
placed, then it maps to an empty location in $\Gamma$.

\paragraph{Simulation overview}
The formation of a (non-seed) supertile $s$ begins from its outside, initiated
by tile attachment from the adjacent sides of existing supertiles that serve
as ``inputs'' to $s$.  The information about each simulated input glue $g$
is conveyed to locations adjacent to the center location of $s$ by a
single tile or pair of tiles specific to the direction of the input and the
strength of $g$.  If the strength of $g$ is $ < \tau$, exactly one tile is
placed with one side adjacent to the center position.  If the strength of $g$
is $ \tau$, two tiles are placed, providing two such sides.  Since $H$ contains
no $\tau$-strength glues, every tile placement must be the result of
cooperation between the glues of at least $2$ tiles, which serve as input glues
to the newly placed tile.  Thus, if $g$'s strength is $< \tau$, it will only
supply one of the necessary sides which can allow a center tile to be placed,
and that will only be able to occur if sufficiently many additional input sides
place tiles adjacent to the center to provide the necessary cooperation.
However, if $g$'s strength is $\tau$, it will be possible for the tiles
representing $g$ to result in the necessary $2$ sides adjacent to the center
position, allowing the center tile to be placed (as would be expected, since in
$\Gamma$ the presence of a single $\tau$-strength glue in a neighboring
position is enough to allow the placement of a tile).  Note that due to the
fact that the center hexagonal tile has only $6$ edges, in order to allow each
input path the potential of acquiring two edges adjacent to the center, the
south and west pair of input sides and the north and east pair each share a
position among the two sides which is ``competed'' for in the case that both
directions represent input sides with $\tau$-strength glues.  Only one of each
pair has the potential to acquire both positions and thus perhaps eventually
claim placement of the center tile.  This is consistent with the simulation of
$\Gamma$ since, if an untiled position has multiple neighbors with
$\tau$-strength glues, any of those neighbors can potentially independently
direct the placement of the new tile.

Finally, after the center position is tiled, any locations adjacent to the center which weren't used as paths for input glues are then tiled by tiles which convey the output glues consistent with the tile $t$ being simulated to the sides which will simulate output sides.

Figure~\ref{fig:hex1a-layout} shows four square tiles arranged around an empty white center square, and the south square tile (yellow) has a strength $s < \tau $ glue $g$ exposed to its north. %
Input glue $g$, and its strength $s$, are encoded at the 3 southmost bold hexagon edges in Figure~\ref{fig:hex1a-layout}, these are the {\em input sides} to the supertile, and each is of strength $\lceil \tau/2 \rceil$. The four sets of hexagonal input sides to the supertile are indicated with bold lines (and arrows) in Figure~\ref{fig:hex1a-layout}. It follows that the supertile's {\em output sides} are as indicated in bold in  Figure~\ref{fig:hex1-outputs}.

The square tile type that is simulated by a supertile is defined by the hex tile type that is placed at the center of the supertile (shown in gray in  Figure~\ref{fig:hex1a-layout}). Sides of supertiles, or {\em supersides}, compete and/or cooperate to claim this center tile location and thus set the identity of the forming supertile (in terms of the square tile type which it is simulating).

\paragraph{Strength $< \tau$}
There are two types of bonds  to simulate, strength $< \tau$ and strength $\tau$. We first consider strength $< \tau$ bonds. Here, input supersides first attempt to place a single hex tile at one of the 4 locations indicated by the colored rectangles in Figure~\ref{fig:hex1a-layout}. Then, the goal is for supersides to cooperate to place a tile at the center gray location. The tile placed at the center defines the square tile type simulated by the supertile. A superside encoding a strength $s$ square glue will advertise a strength $s$ glue to this center tile position as shown in Figure~\ref{fig:hex1a-layout}.  All other glues in the supertile are of strength $\lceil \tau/2\rceil$.

Figures~\ref{fig:hex--NS-less-than-tau} and~\ref{fig:hex--NW-less-than-tau} show two specific examples  where two superside cooperate to place a tile at the center. Numbers indicate precedence of placement within a distance 1 neighborhood, specifically, hex tile 2 can not be placed until the two neighboring tiles numbered 1 have been placed. All hex glues (in the white region) are of strength $\lceil \tau/2 \rceil $, except for two glues: the pair of cooperating tiles with label 1 expose glues of identical strength to their simulated square edge.
The other cases of 2 cooperating supertiles, that encode strength $< \tau$ square sides, are symmetric with Figures~\ref{fig:hex--NS-less-than-tau} and~\ref{fig:hex--NW-less-than-tau}. The case of 3 or 4 cooperating supersides can be understood from Figure~\ref{fig:hex1a-layout}. Specifically, for the case of  4 cooperating supertiles, that encode strength $< \tau$ square sides, the center gray hex tile that encodes the simulated square tile type can be placed via cooperation of all 4 supersides. When there are 3 cooperating supertiles, that each encode strength $< \tau$ square sides, we proceed similarly: 3 hex tiles cooperate to place the center tile, then tiles are cooperatively placed to tile the output hex supersides. %
It is straightforward, but tedious, to see  that everything works in the presence of mismatches: the essential idea is that a mismatching  superside does not cooperate in the placement of the center tile, and furthermore does not block the placement of any output tiles (except at its own superside).

\paragraph{Strength $\tau$}
Simulation of strength $\tau$ bonds is illustrated in Figure~\ref{fig:hex2}. One of the main differences with the $< \tau$ case is that a strength $\tau$ superside must be able to claim the center hex tile without cooperating with other supersides.
Hence the strategy is to tile two positions adjacent to the center position (Figure~\ref{fig:hex-tau-layout}), and let both of these advertise a strength $\lceil \tau/2 \rceil$ glue to the center tile location.  Another trick used here is for strength $\tau$ supersides to share positions where they place their tiles: for example,  in Figure~\ref{fig:hex-tau-layout}) it can be seen that both south and west will try to place a tile at the south-west hex tile location within the supertile. This trick is used to ensure that output paths are not blocked.

\begin{figure*}[tp] %
\centering
\ifabstract
\else
\newcommand{\dwfigspace}{5.5}
\fi
  \subfloat[][Simulating square strength $\leq \tau$ glues using hexagons with strength $< \tau$ glues. Layout: Potential input hexagon sides to the white  $3\times 3$ supertile are highlighted in bold. Input supersides attempt to claim the center tile by first placing two hexagonal tiles as indicated by the small colored rectangles.]{%
        \label{fig:hex-tau-layout}%
        \centering   %
        \hspace{\dwfigspace pt}\includegraphics[height=1.5in]{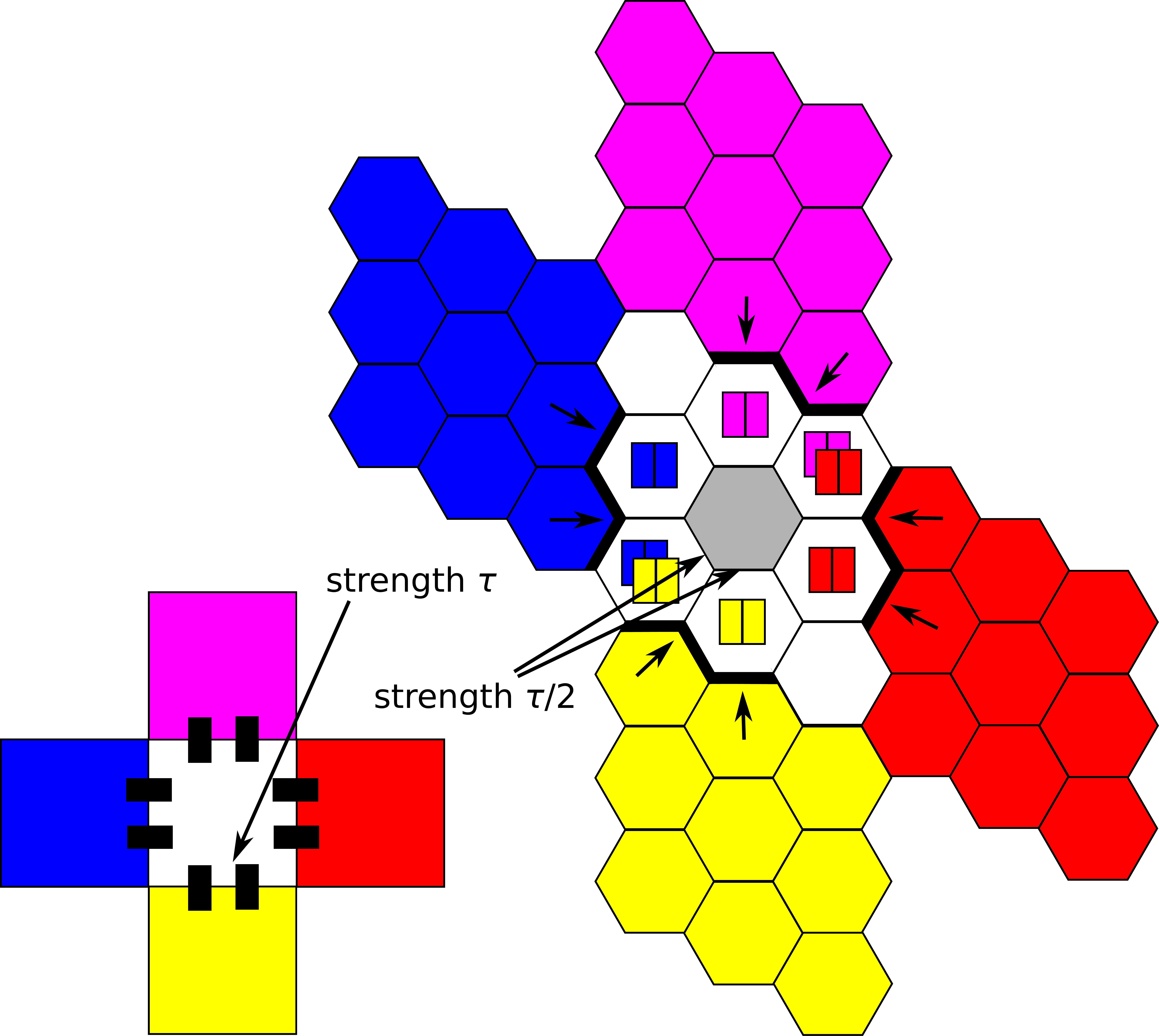}\hspace{\dwfigspace pt}}%
        \hspace{5pt}
  \subfloat[][Simulating a single South strength $\tau$ glue. The center tile is claimed by South placing 3 tiles. Numbers indicate tile placement precedence. 4 subsequent tiles are appropriately placed to encode the output supersides.]{%
        \label{fig:hex-tau-south}%
        \centering
        \hspace{\dwfigspace pt}\includegraphics[height=1.08in]{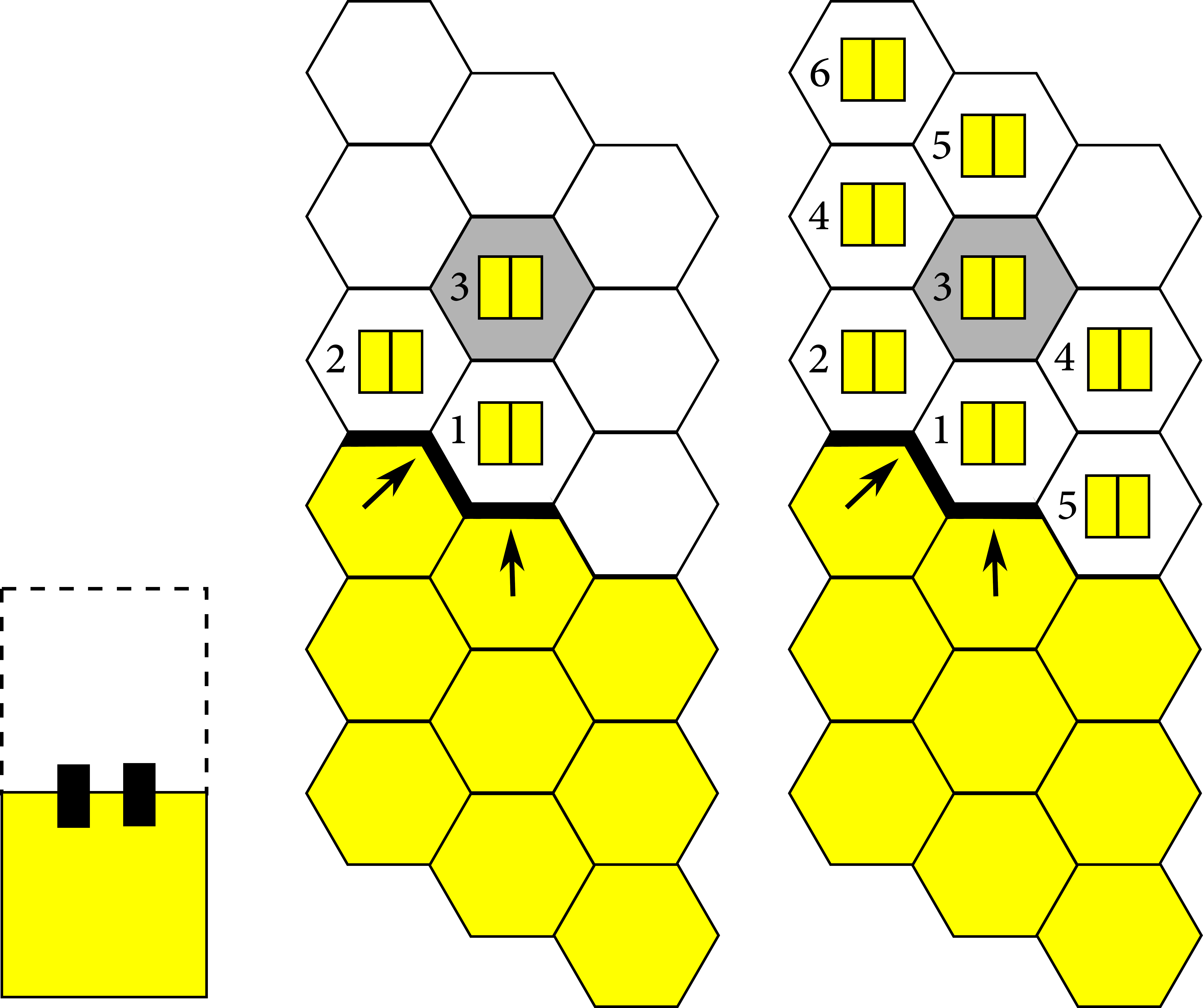}\hspace{\dwfigspace pt}}%
        \hspace{5pt}
  \subfloat[][2 strength $\tau$ supersides compete to claim the center tile. South wins in this case and places an output an superside to the East. In order to place the output superside to the North, South cooperates with one of West's tiles.]{%
        \label{fig:hex-south-tau-c}%
        \centering
        \hspace{\dwfigspace pt}\hspace{10pt}\includegraphics[height=1.3in]{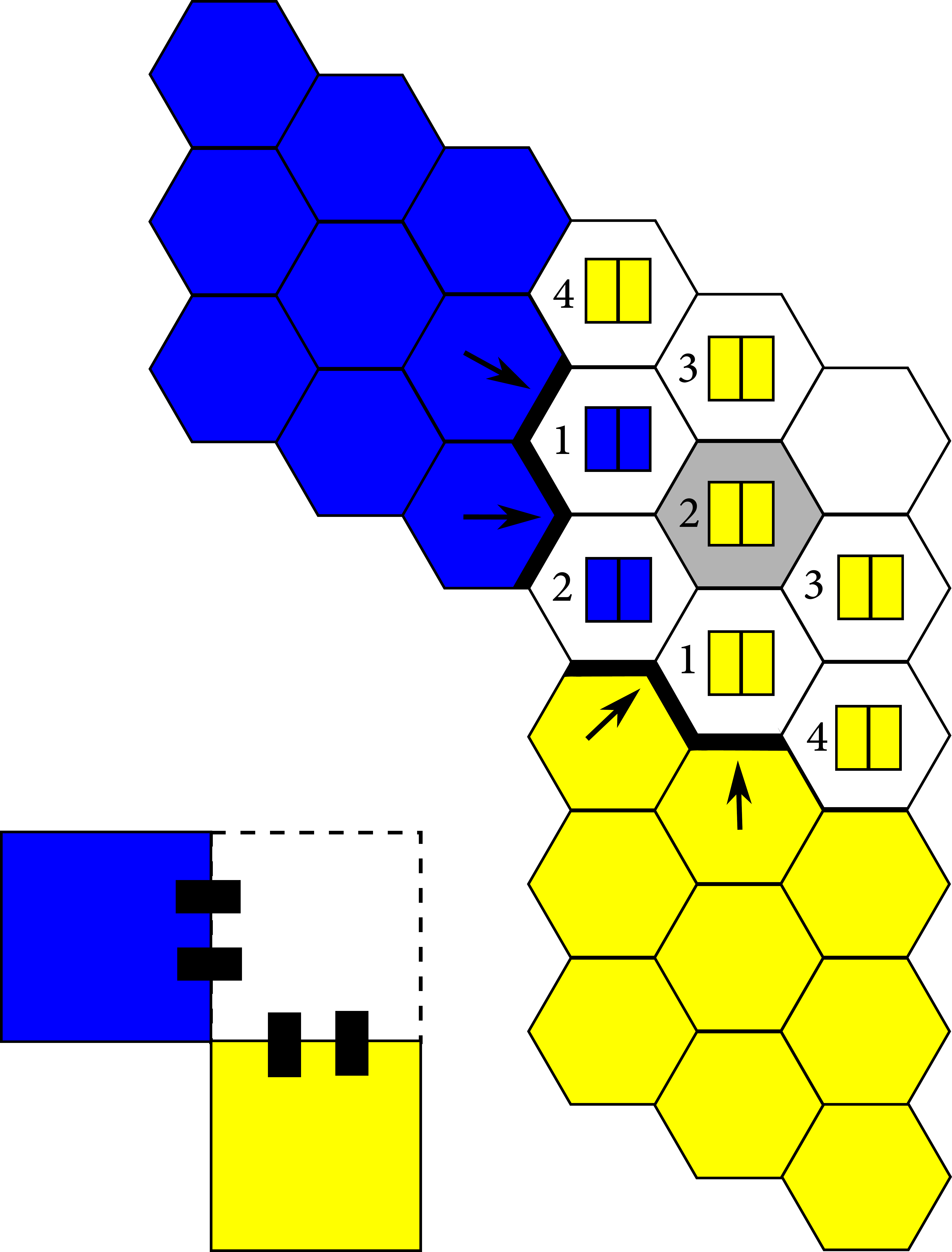}\hspace{10pt}\hspace{\dwfigspace pt}}%
        \hspace{5pt}
  \subfloat[][South wins! A strength $\tau$ south superside claims the center tile in the presence of a North strength $< \tau$ supertile.  South wins and places output supersides to the East and West, by piggybacking on North's single tile.]{%
        \label{fig:hex-tau-south-left}%
        \centering
       \hspace{\dwfigspace pt}\hspace{15pt} \includegraphics[height=1.5in]{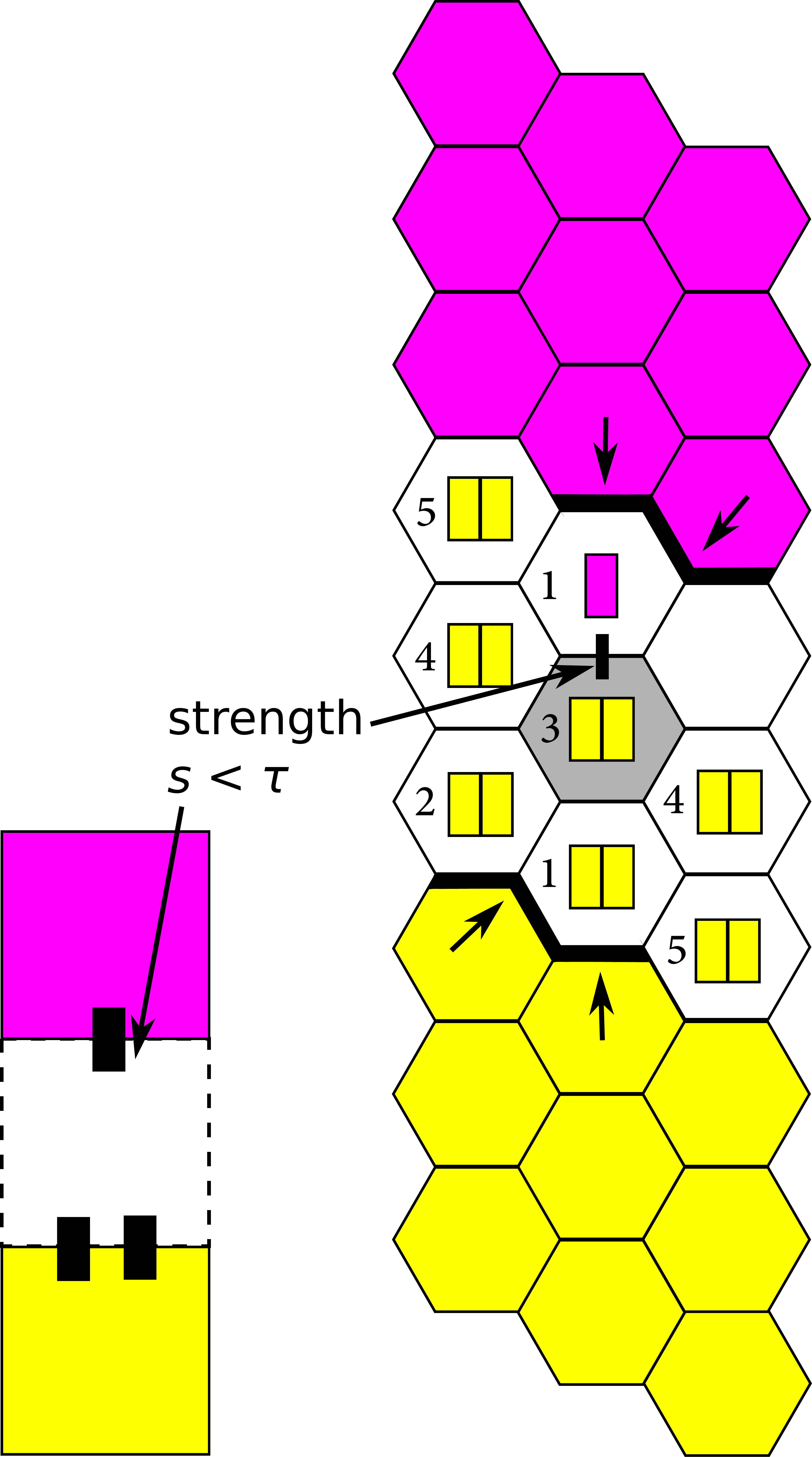}\hspace{\dwfigspace pt}\hspace{15pt}}%
  \centering
  \caption{Simulating square strength $\leq \tau$ glues using hexagons with strength $< \tau$ glues.}\label{fig:hex2}
\end{figure*}

\paragraph{Seed structure}

\ifabstract
\else

Since supertiles are able to grow from existing supertiles in a way which simulates the tile attachments of $\Gamma$, we now need to define the seed structure.  Since $H$ can contain no $\tau$-strength glues, it is impossible for $\Gamma'$ to be singly-seeded, i.e. $|\sigma'| > 1$.  In fact, for $\sigma'$ to be $\tau$-stable, $|\sigma'| \ge 3$.  Therefore, to construct $\sigma'$ we create the hexagonal tiles necessary to form a version of the supertile which represents the single tile of $\sigma$ where all sides are output sides, and create $\sigma'$ by combining 3 of those tiles in a $\tau$-stable configuration.  From this seed, the full supertile will be able to form and the full simulation will be able to proceed.

The seed assembly $S$ of $\Gamma'$ consists of 3 hexagonal tiles that are stable at temperature $\tau$, and that grow into a supertile that encodes the seed tile $s$ of the square system $\Gamma$. This supertile can be hard-coded using 9 unique tile types that use strength  $\lceil \tau/2 \rceil$ glues. Note that  the square system seed tile can appear anywhere in an assembly, and not just at the seed location. Such other versions of the seed supertile do not use these 9 unique tile types, and simply use the standard tiles described in the  construction above. We need to show that a partially formed seed should not be able to initiate growth from one superside, that can grow a sequence of supertiles that come back and block the ``filling out'' of the seed supertile. This can be prevented by having the 3-tile hexagon seed structure consist of a triangle of hexes at the center tile, bottom left and bottom center locations, of the $3\times 3$ supertile. It can seen, by examining Figure~\ref{fig:hex-in-out}(a), that this 3-tile seed must grow to size 5 before it can tile an output superside in such a way that initiates growth of a new supertile. At this point, 2 of the 4 output supersides are tiled, and so they can not be blocked.  Of course, these output sides could produce supertile tentacles that could grow around and attempt to block the other 2 output supersides by tiling inside the 3 ``outputting locations'' of the seed supertile location. For this to happen such a tentacle needs to grow a supertile adjacent to one of the 2 remaining (untiled) supersides, hence that side is blocked in a way that is completely valid. However the tentacle can place 1 or 2 hexagon tiles inside the seed supertile at this point (at input locations), but then its growth is terminated because it can not claim the (already tiled) center tile location. Furthermore, there is one untiled  corner site remaining in the supertile that is adjacent to tentacle tiles. However, this is at an output-only location in the superseed location, and so the tentacle presents 0-strength glues there and so can not tile that position. There remains a tiling path (using both generic, and seed supertile-only glues) to tile the remaining output superside.
\fi %
\ifabstract
\else
\subsection{Correctness of Hexagon Construction}
\begin{figure*}[t]
\centering
        \includegraphics[width=\textwidth]{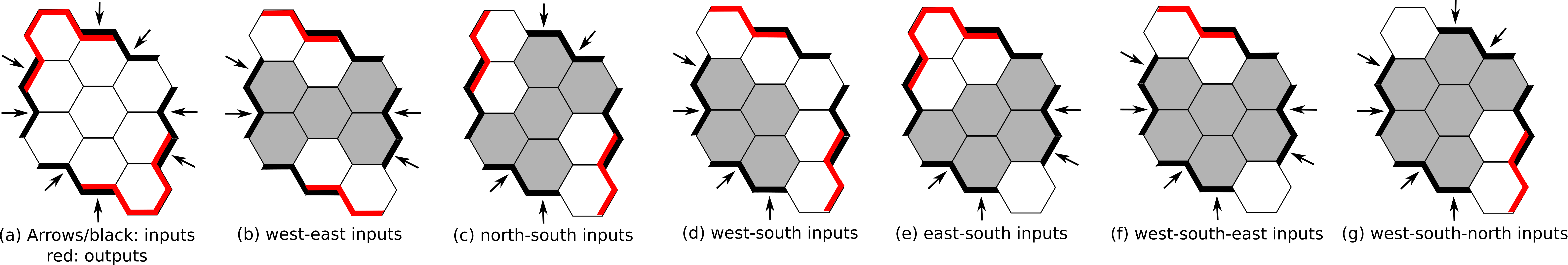}%
  \caption{Hex inputs (bold black, with arrows), and hex outputs (red), for the four supersides. The shaded gray area is a superset of the tiled area upon placement of the center hex tile, irrespective of the strength of square glue being simulated. the figure covers a number of 2- and 3-sided binding cases. In each case we see there is a tiling path to the outputs that can be tiled using strength $\lceil t/2 \rceil$ hex glues.}\label{fig:hex-in-out}
\ifabstract
\vspace*{-3mm}
\else
\vspace*{-5mm}
\fi
\end{figure*}

In order to show that a supertile correctly simulates a square tile, we show that if a square binding event should be simulated, the relevant combination of input supersides leads to a valid placement of output supersides, and otherwise no output supersides should be presented. It can then be seen from Figures~\ref{fig:hex1a-layout} and~\ref{fig:hex1-outputs}, that appropriately tiled output supersides implies that those supersides are ready to act as input supersides to the adjacent supertiles. We consider 4 main cases, where there are 1, 2, 3 or 4 input supersides present (of course, not all input supersides may end up contributing to the supertile choice, e.g.\ 1 or more may mismatch).

When there is exactly 1 strength $<\tau$ superside present the center tile is not claimed, and so no square tile type is simulated  (Figure~\ref{fig:hex1a-layout}). When there is exactly one strength $\tau$ superside present the center hex tile is claimed, as shown in Figure~\ref{fig:hex-tau-south}. There is a path to appropriately place the 3 output supersides (using strength $\lceil \tau/2 \rceil$ glues).

Next we consider the case of 2 input supersides. It can be seen from Figures~\ref{fig:hex--NS-less-than-tau},~\ref{fig:hex--NW-less-than-tau}, \ref{fig:hex-south-tau-c} and~\ref{fig:hex-tau-south-left} (and their rotations) that there are scenarios where either one of the 2 sides can tile a path to the center tile, and the winning superside is determined by the order of tile placement. In other words, each superside has an opportunity to claim the center tile, and no superside is blocked until the center tile is placed.  Figures \ref{fig:hex-in-out}b--e, and all rotations of these, can be used to argue that there is a path to tile the output supersides appropriately. Immediately upon claiming the center tile, the shaded area shows the total area tiled when  there are 2 input supersides of strength $\tau$, and it shows a superset of the tiled are when there are input supersides is of strength $< \tau$. In each case there is a path to tile the remaining 2 output supersides: specifically there is a path to tile the 2 red output locations using strength $\lceil \tau/2 \rceil$ glues.

For the case of 3 input supersides, Figures~\ref{fig:hex1a-layout} and~\ref{fig:hex-tau-layout} can be used to reason that the center tile can be placed, and furthermore that through cooperation between supersides, all supersides have the opportunity to individually claim (strength $\tau$ supersides) or cooperatively claim (strength $< \tau$ supersides) this tile location. Figures \ref{fig:hex-in-out}f and~\ref{fig:hex-in-out}g (and 2 rotations of these) illustrate the situation upon claiming of  the center location (the shaded area represents a superset of the tile locations at this point).  Again, in each case there is a path to tile the remaining output superside, specifically there is a path to tile the red output location using strength $\lceil \tau/2 \rceil$ glues.

For the case of 4 input supersides, Figures~\ref{fig:hex1a-layout} and~\ref{fig:hex-tau-layout}, illustrate the various ways in which the center tile could be claimed.

Of course, it could be the case that the input supersides encode mismatching square sides that correspond to no tile placement, and this is handled by the construction as follows. (a) A superside encodes a glue that has no corresponding opposite side glue: here no hex tiles are at all placed by the corresponding ``input superside'' (the relevant supersides in Figures~\ref{fig:hex1a-layout} and~\ref{fig:hex-tau-layout}  would have {\em no} tiles placed with colored rectangles). (b) superside encoded insufficient strength (or mismatching glues) for an attachment: in this case the tiles denoted using colored rectangles in Figure~\ref{fig:hex1a-layout} advertise insufficient strength, or mismatching glues, to the center tile position.

In order to generate all of the hexagonal tiles that are needed for the different versions of each supertile, it is sufficient to create $O(|T|^2)$ hexagonal tiles.  Intuitively, this allows for the creation of a unique hexagonal tile for each combination of the aTAM tile types represented by its two input tiles.
 \fi  %
\end{proof}

\ifabstract
\else
Furthermore, it is straightforward to see from the above construction that an aTAM system with seed $\sigma$ where $|\sigma| \ge 1$ (i.e.\ a seed assembly consisting of many tiles) is simulated by an hTAM system where the $9|\sigma|$ hexagonal tiles simulating the aTAM seed assembly are appropriately placed to represent that seed assembly. Thus the following Corollary holds.
\fi

\begin{corollary}
For any aTAM system $\Gamma = (T, \tau, \sigma)$ with $|\sigma| \ge 1$ and $\tau > 1$, there exists a hexagon assembly system $\Gamma' = (H, \tau, \sigma')$  that simulates $\Gamma$ and has the property that all glues  in the hexagonal-tile set $H$ are of strength  $< \tau $. Also, $|H| = O(|T|^2)$,  $|\sigma'| = 9|\sigma|$ and the simulation has a constant scale blowup factor of $3 \times 3$. \label{cor:hex}
\end{corollary}

\section{Encoding Glues in Geometry}
\label{sec:glue-geometry}

We utilize the idea (similar to \cite{fu2012SAGT}) of using small surface geometries on polygonal tiles to encode additional information.
These geometries have two canonical types: indentations and protrusions, which we call \emph{dents} and \emph{bumps} for convenience.
Such geometries are also used extensively in subsequent sections to both encode information and enforce constraints on how tiles can bond.
\ifabstract
\else
However, here we merely describe how to perform a simple modification on a polygonal tile system to reduce the number of glues in the system to one by encoding these glues in bump-and-dent geometry.
\fi

\ifabstract
\else

\paragraph{Modification}

Given a polygonal tile set $T$ with glue set $G = \{g_1, g_2, \dots, g_m\}$, we define a \emph{glue geometry} consisting of a sequence of $m$ dents and $m$ bumps.
Each side of every tile in $T$ is modified to include this geometry oriented around the midpoint of the side (dents on the counterclockwise (CCW) side, bumps clockwise (CW)), with the original glue of the side removed and a set of small \emph{subsides} with a common strength-1 glue created within the glue geometry. (From now on, we are working with a slightly generalized pfbTAM model that permits this kind of glue placement.) 
For a side with strength-$k$ glue $g_i$, these subsides are created by dividing the tip of the $i$th leftmost dent and $i$th rightmost bump into $\lfloor k \rfloor$ short sides, each with a strength-1 glue (see Figure~\ref{fig:glue-geom-ex}).
The subsides have non-repeating lengths unique to $g_i$.
If $k$ is odd, a short side between the bumps and dents of a length unique to $g_i$ with the common strength-1 glue is also created.

\begin{figure}[t] %
  \begin{center}
  \includegraphics[width=0.6\columnwidth]{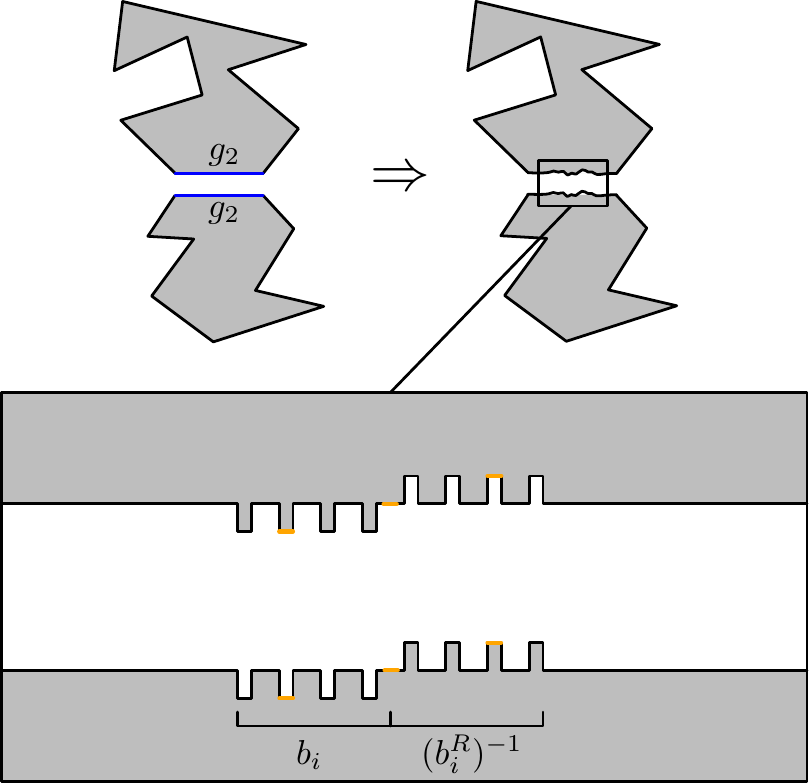}
  \caption{
  \label{fig:glue-geom-ex}
  Converting a side with a strength-5 glue $g_2$ (blue) to a set of smaller sides using a common glue (orange). The original glue set has 4 glues in this case.}
  \end{center}
\end{figure}

\paragraph{Correctness}

First, because the glue geometries are small geometries placed at the midpoint of each side and matching only each other, sides can only attempt to bond with coincident midpoints, as they did in the unmodified tile set, with the $k$ dents of each glue geometry accepting the $k$ bumps of the other.
Also note that side pairs are not prevented from coming together in this way, even if the sides do not bond, as all glue geometries have matching geometries with each other.

Next, any pair of unmodified sides that bond with positive strength have identical glues.
In the modified tile set, the glue geometries of these sets both have a sequence of subsides with strength-1 glues on the tips of the $i$th leftmost dent and $i$th rightmost bump (and possibly on a short side spanning the midpoint of the unmodified edge).
As a result, a total of $2 \cdot \lfloor k/2 \rfloor$ strength-1 bonds are formed, plus 1 if $k$ is odd.
So bonds with total strength $k$ are formed in the modified tile set.

Finally, any pair of sides in the original tile set with distinct glues do not form a positive strength bond, but may become coincident and form a ``strength-0 bond''.
In this case, the glue geometries have matching geometries as mentioned before, but the glues in the glue geometries for both sides lie on distinct bumps and dents, and do not meet.

All that remains is to show that no unwanted positive-strength bonds can occur in the modified tile set.
The cases in which the two corresponding sides in the original tile set are coincident has already been considered.
Recall that the lengths of the subsides are unique to each $g_i$.
As a result, the only sets of subsides that can form positive-strength bonds are those corresponding to the same $g_i$.
Because the lengths are non-repeating, the only matching sequence of subsides occurs when two bumps both have subsides with glues corresponding to the same glues \emph{and} these glues appear mirrored, as when one tile has been reflected or ``flipped''.

Since the glue geometry is small and the constructions using glue geometry either use large bump and dent geometry to forbid tile reflection (as in Section~\ref{sec:manygons}) or exist in models that forbid reflection, such an orientation of subsides cannot occur.
So all bonds occuring in the modified tile set correspond to valid bonds in the original tile set with the same total strength.

\fi

\section{Self-Assembly With a Single, Rotatable Tile Type}
\label{sec:manygons}

Now we introduce the use of \emph{rotation} of a single polygonal tile to encode a set of multiple simple tiles.
We show that hTAM systems that do not contain $\tau$-strength glues can be simulated by a pfbTAM system consisting of a single (nearly) regular convex polygon (with small surface geometry) that may rotate and flip.
The hTAM systems are simulated at scale 1 in a natural way: the rotation angle of a tile indicates the hexagonal tile it simulates.
The use of small surface geometry (as used in Section~\ref{sec:glue-geometry}) constrains the possible orientations in which the tile can attach to the seeded assembly to match only
\ifabstract
\else
those that correspond
\fi
to valid hexagon tiles.

\ifabstract
\else
First we define the class of polygonal tiles utilized for this construction.
Then we present the full details of the construction for simulating any hTAM system without any $\tau$-strength glues using a single-tile pfbTAM system.
Because there are no $\tau$-strength glues in the simulated hTAM system, any such system assembling a non-trivial assembly requires a multi-tile seed.
The construction we present uses an equivalent seed in the simulating pfbTAM.
We then use this construction to prove Theorem~\ref{thm:single-tile-iusa}.
\fi

\ifabstract
\else
Finally, we describe a modification to the construction that eliminates the need for a multi-tile seed in the simulating pfbTAM (by using a $\tau$-strength glue), thereby making the pfbTAM system self-seeding.
In this case no seed assembly must be created as input for the system, but instead one individual copy of the single tile type will assume the role of seed for each assembly which forms by then growing from it), and thus prove Theorem~\ref{thm:self-seed-sim2}.
\fi

\subsection{$n$-gon tiles}\label{sec:n-gons}

\ifabstract
\else
A \emph{regular $n$-gon} is an equilateral and equiangular convex polygon with $n$ sides, for $n \in \mathbb{N}$ and $n \ge 3$ with each side length 1.
\fi
 For this construction, we define a family of polygonal tiles we call \emph{$n$-gon tiles} as a set of unit regular $n$-gons with the addition of small-scale surface protrusions (\emph{bumps}) and indentations (\emph{dents}), as used to encode glues in Section~\ref{sec:glue-geometry}.
\paragraph{Side geometry}

The geometry of each side of a tile consists of the placement of a bump and dent pair in one of two locations on the polygon side, either towards the clockwise or counterclockwise end of the side.
\ifabstract
\else
See Figure~\ref{fig:manygon-b-n-d-details} for an example of the geometries on a polygonal tile and details about the sizes and placement of bumps and dents.
The interior angles between the edges of the $n$-gon are $\alpha = (1 - 2/n) \cdot 180^\circ$.
Each bump consists of a rectangular protrusion, extending perpendicularly from the tile's side, while each dent is a rectangular concavity in the side.
Bumps and dents have identical size and are placed adjacent to each other.
Let $h$ and $w$ be the height and width of the bumps and dents, and $d$ be the distance separating them, with dents preceding bumps when travelling counterclockwise along the boundary of the tile.
Set $h = d  \cdot  \tan(180^{\circ} - \alpha) / 2$, $w = d/4$, $d = 1/10$, where $CCW$ and $CW$ geometries are placed distance $d$ from the counterclockwise-most and clockwise-most ends of the polygon edge, respectively.

\fi
\ifabstract
\else
We claim that these geometries have the following properties: 1. dents in the interior of the tile do not collide with each other (see the dents in the detail of Figure~\ref{fig:manygon-b-n-d-details} for an example of dents that are close to a collision), and 2. for any two tiles with a pair of coinciding sides, the bump on any non-coinciding side of a polygon does not intersect the other polygon.

Property 1 holds as long as the dents in consecutive sides with $CCW$ and $CW$ geometries do not collide, because all other pairs of dents are further apart.
Because the fewest sides that an $n$-gon could have is $6$ (which would occur if $c = 1$), the smallest that $\alpha$ can be is $(1 - 2/6) \cdot 180^\circ = 120^\circ$.
Then because $h = d \cdot \tan(180^{\circ} - \alpha) / 2 \leq d \cdot 1.74 / 2 < d$, it is impossible for two dents to intersect.

Property 2 holds provided that a $CW$ bump on side $s_i$ is contained entirely within the half-plane formed by the supporting line (the dotted line in Figure~\ref{fig:manygon-b-n-d-details}) of side $s_{i-1}$ and containing the rest of the tile.
Because $\alpha \geq 120^{\circ}$ for all $n$-gon tiles with $n \geq 6$, then $\tan(180^{\circ}-\alpha) \leq 1.74$.
The bump is distance $d$ away from the vertex and has height $h$, so since $h/d = \tan(180^{\circ} - \alpha)$, the bump does not intersect the supporting line.

\begin{figure}[t] %
  \begin{center}
  \includegraphics[width=0.6\columnwidth]{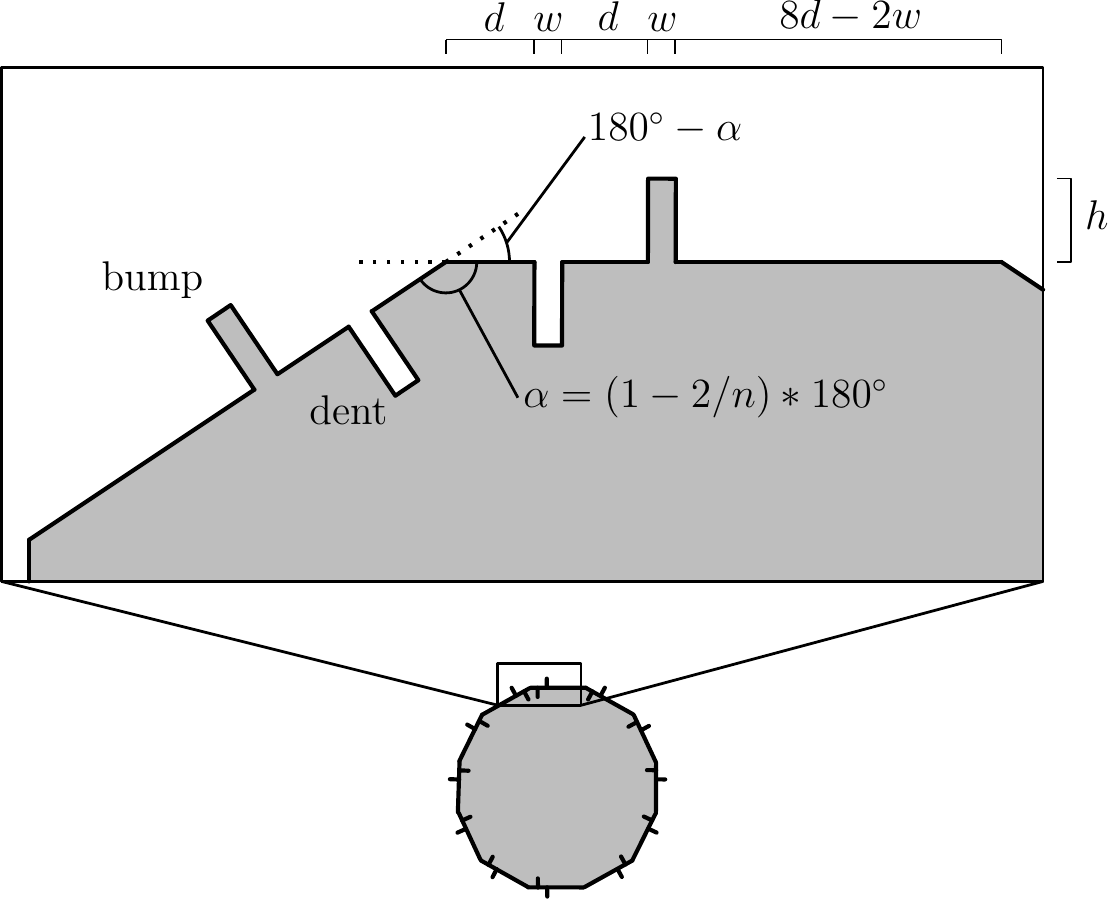}
  \caption{
  \label{fig:manygon-b-n-d-details}
  A pair of adjacent sides of an $n$-gon tile.}
  \end{center}
\end{figure}
\fi

Each $n$-gon tile $t \in F$ has a sequence of sides oriented CCW and starting with $s_0$, and an identity rotation and flip with $s_0$ horizontal at the top of the polygon.
Specifying each side of $t$ requires a glue label and location for the bump-dent geometry (either at the $CW$ or $CCW$ end of the side).

\subsection{Simulating the hTAM}\label{sec:many-gons}

\begin{lemma}\label{lemma:poly-sim}
For any hTAM system $\Gamma_h = (T, \tau, s)$ such that $T$ does not contain any $\tau$-strength glues and $|s|=3$, there is a free-body self-assembly system $\Gamma_f = (F, \tau, s_f)$ with $|F| = 1$ and $|s_f|=3$ that simulates $\Gamma_h$.
\end{lemma}

\begin{proof} (Lemma~\ref{lemma:poly-sim})
To prove Lemma~\ref{lemma:poly-sim}, we provide the following construction to create $\Gamma_f$ from the definition of $\Gamma_h$.  Let $c = |T|$, the number of tile types in the simulated system.  $F$ will consist of a single $n$-gon tile type where $n = 6c$, created by utilizing the following transformations:

\ifabstract
\else
\begin{figure}[htp]
\centering
  \subfloat[][Assuming $T$ consists of $2$ tile types, the blue and green, the $12$-gon tile on the right would be created.]{%
        \label{fig:polygonal-tile-creation1}%
        \centering
        \includegraphics[width=.9\columnwidth]{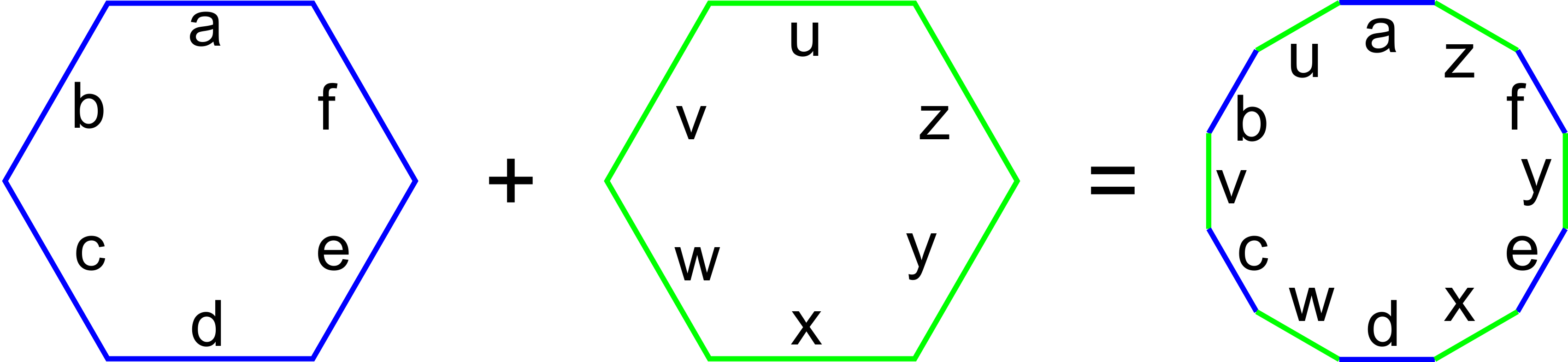}}%
        \hspace{25pt}
  \subfloat[][The $12$-gon inscribed in a hexagon, with the orientation which represents the blue tile (left), and the orientation representing the green tile (right).]{%
        \label{fig:polygonal-tile-creation2}%
        \centering
        \includegraphics[width=.6\columnwidth]{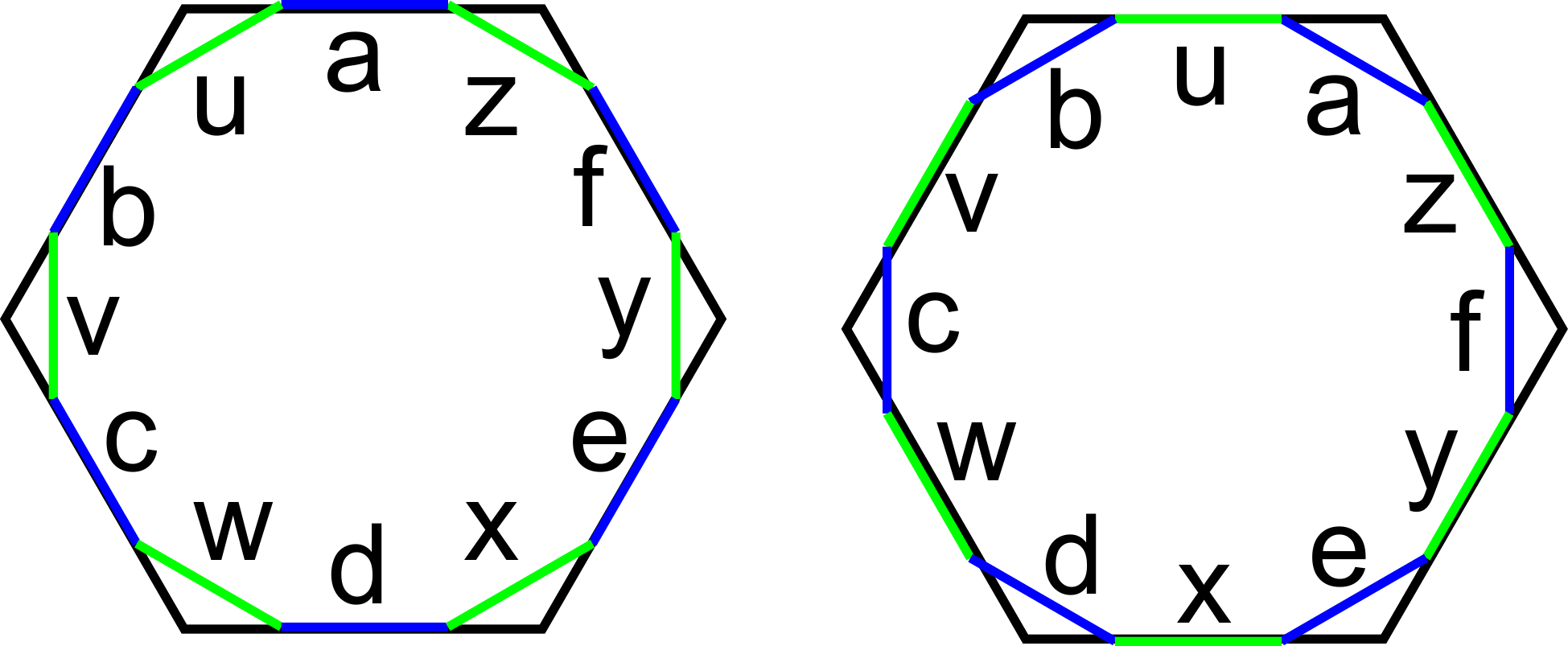}
        }%
  \caption{Example creation of a $12$-gon tile type.}
  \label{fig:polygonal-creation}
  \centering
\end{figure}
\fi

\begin{enumerate}

\item
\label{item-direction-pairs}
Create a new (hexagon) tile set $T'$ as follows: for each $t \in T$, add a new tile type $t'$ in $T'$ where every glue is also labeled according to the pair of opposite sides ($1$ for $N/S$, $2$ for $NW/SE$, $3$ for $SW/NE$) it is found on.
\ifabstract
\else
For instance, a set of six counterclockwise glues $a, b, a, c, d, b$ on $t$ would become $a_1, b_2, a_3, c_1, d_2, b_3$.
This ensures that a glue appearing on a side $s$ of a tile is only able to bind to another glue which is on either the same or opposite side of a tile.
\fi

\item
Let $p \in F$ be an $n$-gon tile type.
Orient $p$ in its identity orientation, with $s_0$ placed horizontally as the uppermost side.
For sides $s_j, 0 \leq j < c$, assign the North glue of the $j$th tile in $T$.
Similarly, for sides $s_{j+c}, s_{j+2c}, \ldots, s_{j+5c}$, assign the NW, SW, S, SE, and NE glues of the $j$th tile in $T$.
\ifabstract
\else
Thus, each consecutive group of $n/6$ sides, starting from $s_0$, represents the sides of one particular direction of the $c = n/6$ tiles in $T$.
This results in the sides of $p$ containing the $6$ glues of each tile type in $T$ in positions such that they coincide with the $6$ sides of the smallest hexagon into which $p$ can be inscribed (in some orientation).
See Figure~\ref{fig:polygonal-creation} for an example.
\fi

\item
The final modification to $p$ adds geometries consisting of an adjacent bump and dent pair per side at one of two locations ($CCW$ or $CW$) along the edge.
Assign the geometry locations for $s_j, j < 3c$ to be $CCW$, and the geometry locations of for $s_j, 3c \leq j$ to be $CW$.
\ifabstract
\else
Thus the geometry on the consecutive sides representing the North, Northwest, and Southwest sides of tiles in $T$ has one configuration, and the geometry on the other sides is distinct.
\fi

\end{enumerate}

To form the seed $s_f$, we define $s_1$, $s_2$, and $s_3$, where $s_1,s_2,s_3
\in T$, as the three tile types that form $s$.  We then take $3$ copies of the
$n$-gon $p \in F$, denoted by $s_{p1}$, $s_{p2}$, and $s_{p3}$.  Orient
$s_{p1}$ so that if it is inscribed within the smallest hexagon that contains
it, and that hexagon is oriented such that it has north and south sides, then
the $6$ sides of $s_{p1}$ that coincide with the bounding hexagon represent
the $6$ sides of $s_1$ in that same orientation.  Orient $s_{p2}$ and $s_{p3}$
similarly with respect to $s_2$ and $s_3$.  Now place $s_{p1}$, $s_{p2}$, and
$s_{p3}$ adjacent to each other (without additional rotation) in the same
relative configuration as $s_1$, $s_2$, and $s_3$ occur within $s$.  Since the
three adjacent pairs of edges in $s$ bind to form a $\tau$-stable assembly, by
the construction of $p$, the edges of the $3$ copies of $p$ in their current
rotations and positions are guaranteed to also bind to form the analogous
$\tau$-stable seed $s_f$.

Assembly proceeds by correctly oriented copies of $p$ binding with at least $2$ edges of $n$-gons in the growing assembly (since there are no $\tau$-strength glues).  As it has been shown that correct tile bindings will occur, we now show that incorrect bindings cannot occur.  Incorrect bindings must be prevented in the cases of 1. rotations that don't correspond to hexagonal tiles which should bind being allowed to bind, and 2. copies of $p$ which have flipped over (and possibly also rotated) being allowed to bind.

To correctly simulate non-rotatable hexagons, each side of $p$ must only bind to other sides that simulate the opposite side of a hexagon, i.e., sides of $p$ that represent N, NW, SW, S, SE, and NE sides of tiles from $T$ must bind only to sides of $p$ that represent S, SE, NE, N, NW, and SWsides, respectively.
As previously mentioned, this creates the pairs of opposite directions ($N/S$, $NW/SE$, and $SW/NE$).
Binding to a side representing a direction not in the correct pair is ensured by Step~\ref{item-direction-pairs} in the creation of $p$.
Correct binding to a complementary side is allowed by the existence of the correct glues and also the definition of the geometries.

Incorrect binding is prevented by a combination of geometries and direction-specific glues.
For each orientation corresponding to a hTAM tile, there are 11 others that are ``incorrect'', but related to the tile: they are the 11 orientations that correspond to rotations of flips of the hexagonal tile.
Because the simulating pfbTAM system uses distinct glues for each pair of opposite directions, and at least two bonds are needed for a tile to attach, there are at most four possible orientations of the 12 total that can possibly attach: 1. the valid orientation, 2. a $180^{\circ}$ rotation from valid, 3./4. a flip about the direction pair containing the two bonds, and a $180^{\circ}$ rotation of this flip.
Note that 3. and 4. only exist if only two bonds are needed to attach and they lie in the same direction pair, otherwise only 1. and 2. are possible.
Figures~\ref{fig:polygonal-tile-bad-rotation}, \ref{fig:polygonal-tile-bad-flip1}, and~\ref{fig:polygonal-tile-bad-flip2} show that Cases~2,~3, and~4 are all forbidden due to mismatch in geometries.
\begin{figure}[htp]
\centering
  \subfloat[][Two copies of $p$, with top rotated $180^\circ$ relative to bottom.]{%
        \label{fig:polygonal-tile-bad-rotation}%
        \centering
        \hspace{3pt}\includegraphics[width=0.8in]{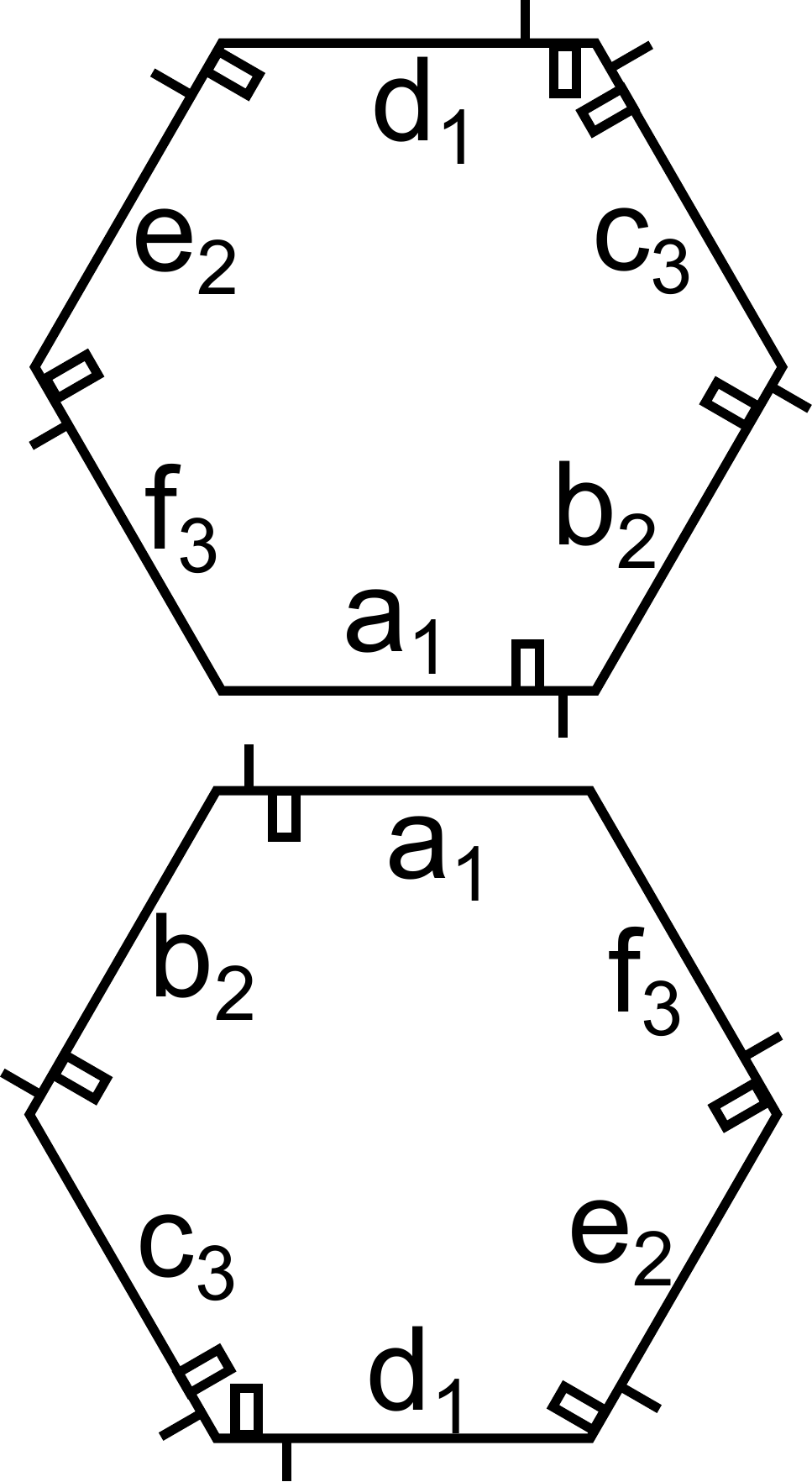}\hspace{6pt}}%
        \hspace{7pt}
  \subfloat[][Two copies of $p$, with top flipped
\ifabstract
\else
horizontally
\fi
and  rotated $180^{\circ}$.]{%
        \label{fig:polygonal-tile-bad-flip1}%
        \centering
        \hspace{2pt}\includegraphics[width=0.8in]{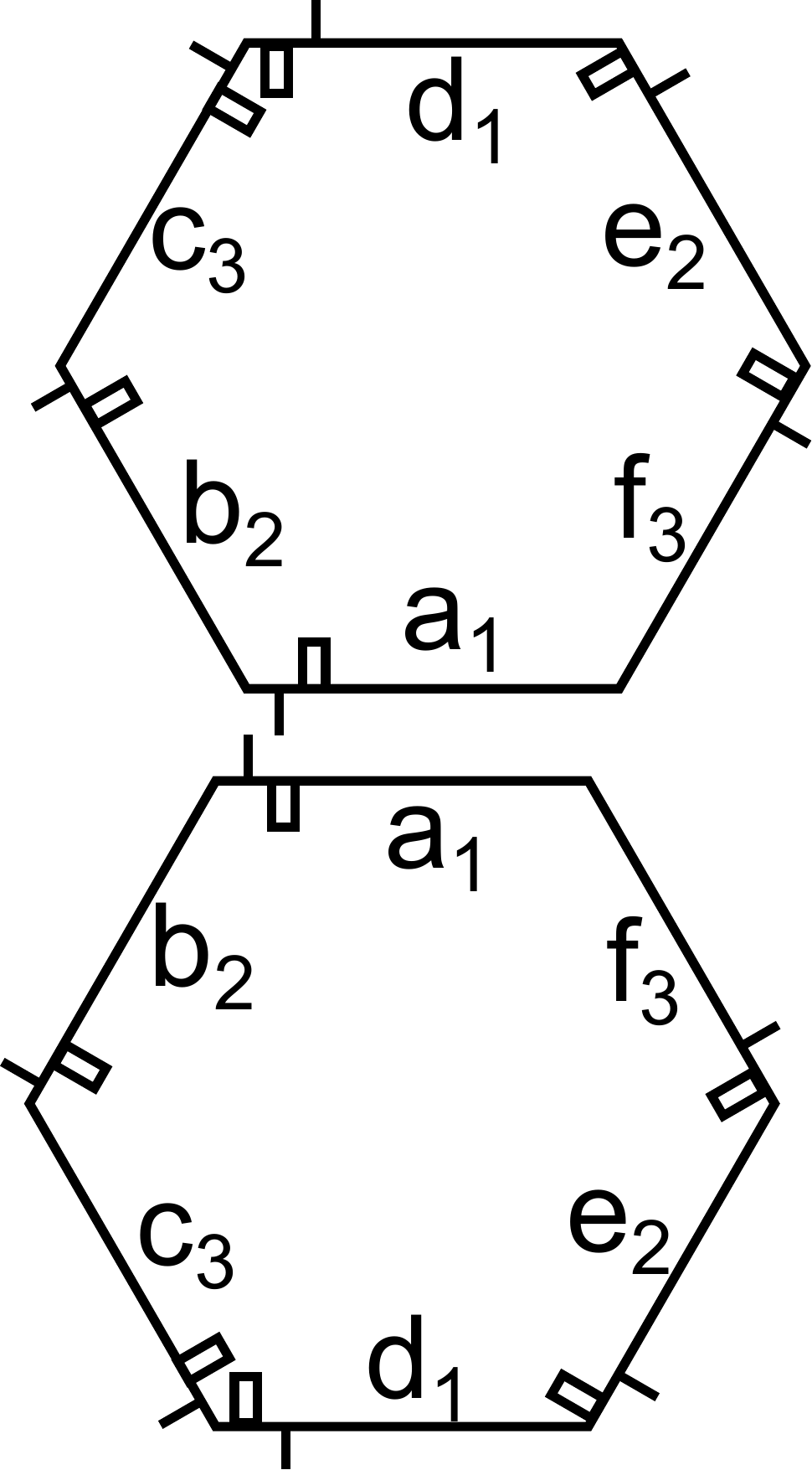}\hspace{2pt}
        }%
        \hspace{7pt}
  \subfloat[][Two copies of $p$, with top flipped horizontally.]{%
        \label{fig:polygonal-tile-bad-flip2}%
        \centering
        \hspace{2pt}\includegraphics[width=0.8in]{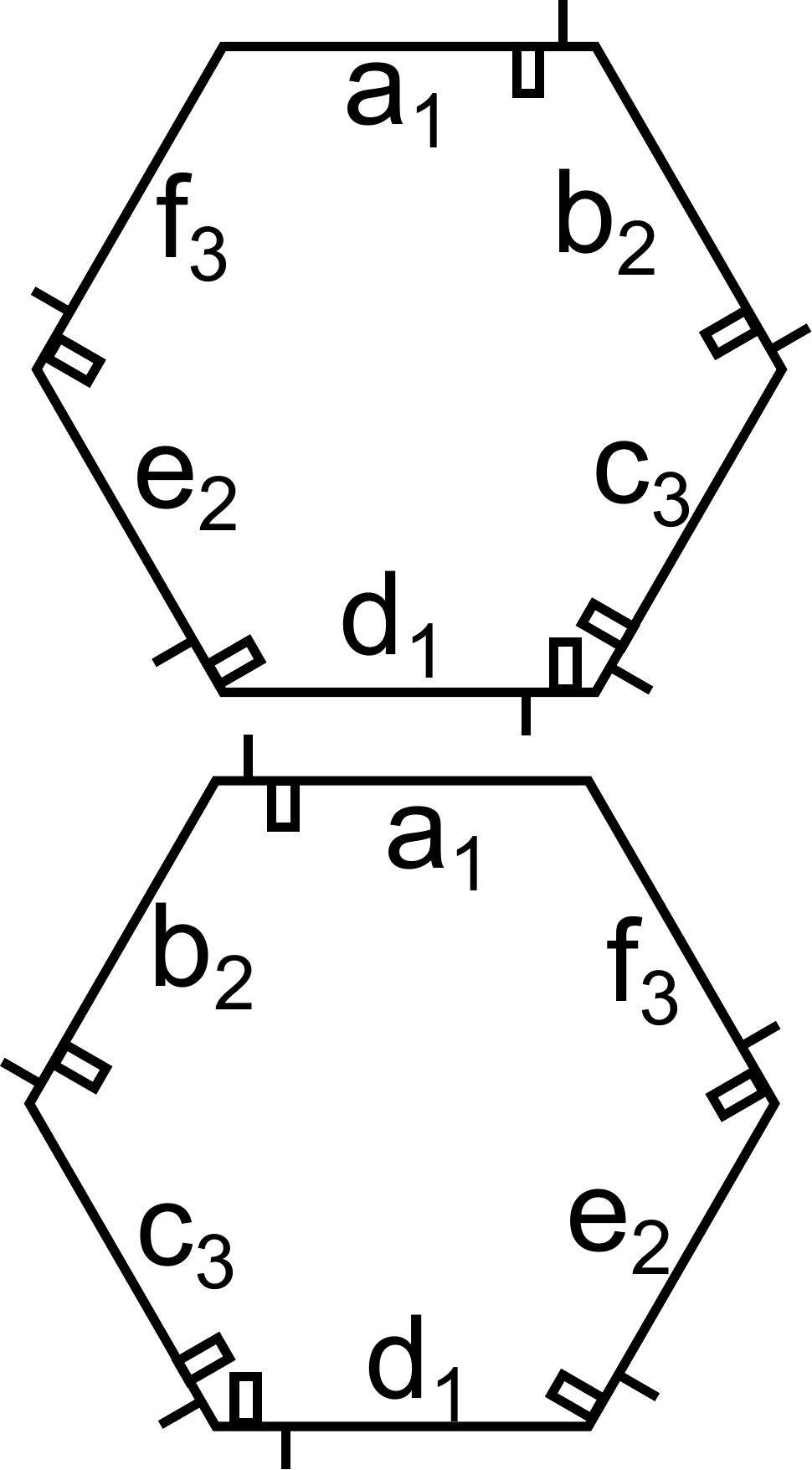}\hspace{2pt}
        }%
  \caption{Cases with glues in correct opposite-direction pairs, but binding is prohibited by geometry mismatches.}
  \label{fig:polygonal-bad-bindings}
  \centering
\ifabstract
\else
\vspace*{6mm}
\fi
\end{figure}
\end{proof}

\begin{theorem} (Universal Single Tile Simulation)\label{thm:single-tile-iusa}
There exists a single polygonal tile $p$ such that, for any aTAM system $\Gamma = (T,\tau,\sigma)$ where $|\sigma| = 1$ and $\tau > 1$, there is a $\tau$-stable seed assembly $\sigma'$ such that the pfbTAM system $\Gamma' = (\{ p\}, 2, \sigma')$ simulates $\Gamma$.
\end{theorem}

\begin{proof}%
By Theorem 3.1 of \cite{IUSA}, there is a single aTAM tile set which, when properly seeded and operating at temperature $\tau = 2$, can simulate any aTAM tile assembly system.  Let $U$ be that tile set.  Let $\mathcal{T}$ be an arbitrary aTAM tile assembly system.  By Theorem 3.1 of \cite{IUSA} there is an aTAM system $\mathcal{T'} = (U, 2, \sigma)$ which simulates $\mathcal{T}$.  Let $\mathcal{H} = (H, 2, \sigma_H)$ be the hTAM system which simulates $\mathcal{T'}$ without using strength 2 glues (as constructed using the construction in Section~\ref{sec:hex} and defining the seed $\sigma_H$ as the fully formed supertiles which represent $\sigma$ as mentioned for Corollary~\ref{cor:hex}).  Let $p$ be the polygonal tile type generated from $H$ using the steps in the proof of Lemma~\ref{lemma:poly-sim}, and $\sigma_p$ be a seed assembly of properly rotated and positioned copies of $p$ to represent $\sigma_H$. Then, it follows that $\mathcal{P} = (\{p\}, 2, \sigma_p)$ simulates $\mathcal{T}$.
\end{proof}

As $U$ is a fixed tile set, $p$ is also a fixed polygonal tile with a
constant number of sides.  With $T$ tiles in the simulated
aTAM system, the scale factor of our simulation is $O(|T|^4log|T|)$, as that
is the scaling factor for the simulation of $\mathcal{T}$ by $\mathcal{T'}$
(see~\cite{IUSA}), and the only further scaling is a constant (i.e.~$3\times
3$) for the simulation of $\mathcal{T'}$ by $\mathcal{H}$.

\ifabstract
\else
\subsection{A Self-Seeding pfbTAM}

In the construction described in~\ref{sec:many-gons}, the three-hexagon-tile seed assembly in an hTAM is simulated by a three-polygon-tile seed assembly in the pfbTAM.
In this section we describe an extension of this construction that also simulates any hTAM system with no strength-$\tau$ glues, but does not use a multi-polygon-tile seed.
Instead it has the more traditional single-tile seed, with the tradeoff that a single strength-$\tau$ bond is used.
Notice that if no $\tau$-strength bonds exist in a pfbTAM system (or any seeded aTAM-like system), assembling any multi-tile assembly requires a multi-tile seed, as no two-tile $\tau$-stable assemblies exist.

The construction will be similar to the simulation construction in that the tile is a polygon tile with glues on each side and small geometries.
However, several additional steps are used to create a polygon tile that simulates the hexagon system and has the following property: every new polygon tile that attaches to the seed assembly during the assembly process bonds using a side that lies in a particular $60^{\circ}$ wedge of the polygon tile.
To achieve this, the techniques of \cite{Versus} for their simulation results of Section 3, specifically ``minimal glue sets'' and ``inward-outward glues'', are utilized.

\begin{theorem}(Self-Seeding Single Tile Simulation)
\label{thm:self-seed-sim2}
For any aTAM system $\Gamma = (T,\tau,\sigma)$ with $|\sigma| = 1$ and $\tau > 1$, there is a pfbTAM system $\Gamma' = (F, \tau, \sigma')$ with $|F| = 1$ and $|\sigma'| = 1$ that simulates $\Gamma$.
\end{theorem}
\fi

\ifabstract
\else
\paragraph{The io-hTAM}

Let $\Gamma_h = (T, \tau, s), |s| = 3$ be the hTAM system being simulated.
For each hexagon tile $t \in T$, define a \emph{minimal glue set} to be a subset of the sides of $t$ whose glues have total strength at least $\tau$ such that any strict subset of these sides has total glue strength less than $\tau$.
Create a second hTAM system $\Gamma_h' = (T', \tau, s'), |s'| = 3$ where for each tile in $T$, a set of tiles in $T'$ are created, one for each minimal glue set.
For each tile created, the glue on each side in the minimal glue set is marked $IN$, and the glues of all other sides are marked $OUT$.
We define the io-hTAM model to be identical to the hTAM model except that assembly proceeds under the following additional constraints: no pair of $IN$ sides from distinct tiles may be adjacent, and any pair of adjacent $OUT$ sides form a strength-$0$ bond (but are permitted to touch).
Add three additional tiles to $T'$ generated from the three tiles in $s$ by marking one side per tile $IN$ and all others $OUT$ such that the these three tiles when placed in the same configuration as $s$ form a $\tau$-stable three-tile assembly.
Define this configuration of these three tiles as $s'$.
Define two assemblies $\alpha, \alpha'$ producible by $\Gamma_h, \Gamma_h'$ to be \emph{equivalent} if replacing each tile $\alpha'$ with the tile generating it in $\Gamma_h$ yields $\alpha$.

\begin{lemma}
\label{lem:mod-htam-equiv-and-out}
An assembly is producible by $\Gamma_h$ in the hTAM model if and only if there is an equivalent assembly producible by $\Gamma_h'$ in the io-hTAM model.
Moreover, every exposed side on every assembly produced by $\Gamma_h'$ in the io-hTAM model is marked $OUT$.
\end{lemma}

\begin{proof}
The proof is by induction.
Recall that the seed $s'$ is constructed such that all exposed sides are marked $OUT$.
Now consider producing an assembly $\beta'$ in $\Gamma_h'$ by adding a tile $t'$ to an existing assembly $\alpha'$ where 1. $\beta'$ has an equivalent assembly $\beta$ producible by $\Gamma_h$ by adding a tile $t$, and 2. all exposed sides are marked $OUT$.

Such a tile $t'$ exists, as the set of bonds used by $t$ to attach to $\alpha$ contains a minimal glue set also sufficient to bind and $T'$ contains a tile where this set of sides is marked $IN$ (and all others are marked $OUT$).
Moreover, any tile $t' \in T'$ was generated by some tile in $T$ and if $t'$ can bind in the io-hTAM, this tile must be able to bind in the hTAM, as its binding requirements are strictly weaker.
Additionally, any tile $t'$ attaching uses all $IN$ sides to bond, so any exposed side on $t'$ after attaching to $\alpha'$ to form $\beta'$ must be $OUT$.
So all sides of $\beta'$ are marked $OUT$.
\end{proof}

\paragraph{Implementing the io-hTAM in the pfbTAM}

Next, we consider implementing an io-hTAM system as a pfbTAM system with a single polygon tile.
To do so, we follow the approach in the previous simulation construction with the addition of new geometry to each side to enforce the $IN$-$OUT$ side constraints.
Recall that we want two sides with $IN$/$OUT$ markings to attach using the following rules:

\begin{itemize}
\item Two $IN$ sides are not permitted to meet (forbidden).
\item A pair with opposite marks can meet and bond according to glue strength (matching).
\item Two $OUT$ sides can meet but not form a positive strength bond (permitted).
\end{itemize}

Implementing this as geometry is simple: an $IN$-marked side has a small bump, while an $OUT$-marked side has a matching dent.
See Figure~\ref{fig:manygon-b-n-d-extension-details} for a detailed depiction of the geometry.

\begin{figure}[htp]
    \begin{center}
    \includegraphics[width=0.6\columnwidth]{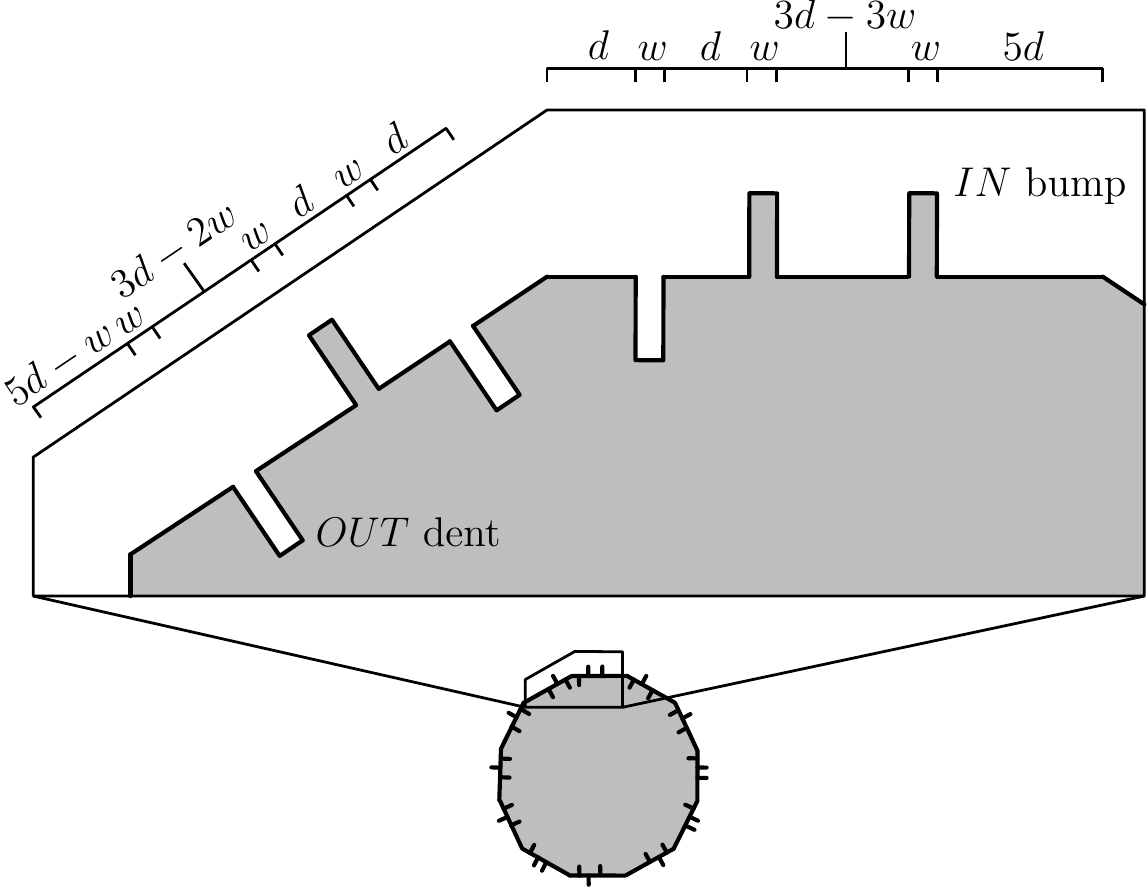}
    \caption{
    \label{fig:manygon-b-n-d-extension-details}
    A pair of adjacent sides of an $n$-gon tile simulating a tile system in the io-hTAM model.}
    \end{center}
\end{figure}

In this case two sides with $IN$ cannot meet due to geometry, $IN$/$OUT$ pairs of sides can meet and can bond, and two sides with $OUT$ can meet, but do not have matching geometry and thus do not form a bond.
We place the bump or dent at the bisector of each side, with the height (or depth) and width equal to that of the bump and dent system implemented in the previous construction.
Now construct a single-tile pfbTAM system as done previously, but with the modified hTAM and additional geometry on each side: for each hexagon tile with marked $IN$ and $OUT$ sides, add six sides to the polygon at $60^{\circ}$ intervals.
Each side corresponds to a side of the hexagon tile, with appropriate glue, bump-and-dent pair according to its location on the hexagon, and bump or dent according to whether it is an $IN$ or $OUT$ side.

\paragraph{Enforcing a polygon tile invariant}

The previous steps yield a single-tile pfbTAM system that simulates an arbitrary io-hTAM system with no strength-$\tau$ glues.
Next, the polygon tile is modified so that during assembly \emph{every} attaching tile must use a side that lives within a particular $60^{\circ}$ arc of the tile's boundary.
Recall that each hexagon tile in the io-hTAM system must use all $IN$ sides (forming a minimal glue set) to attach to an existing assembly, and each hexagon tile in the io-hTAM corresponds to a set of six sides of the polygon tile spaced at $60^{\circ}$ intervals.
For each such set of six sides, `rotate'\footnote{Here rotation is equivalent to cyclic permutation.} the set until the side lying in the interval $[0, 60^{\circ})$ is one which is marked $IN$.
Now the polygonal tile still may attach to simulate the hexagon tiles as before, but with the invariant that every attachment includes a bond formed by a side in the $[0, 60^{\circ})$ interval arc on the boundary of the polygon.

\paragraph{Adding self-seeding}

All that remains is a modification to the simulating polygonal tile that removes the requirement of starting with a three-tile assembly. 
Note that during the assembly process, every attaching tile uses two sides to bond, and any side within $60^{\circ}$ of either of these sides lies in a narrow region between the two tiles sharing the bond.
So geometry prevents polygonal tiles from bonding to such a side.
Additionally, because each attaching polygonal tile uses a side in its angular wedge $[0^{\circ}, 60^{\circ})$, no polygonal tile uses more than one side from this set. 
We now use this geometric fact to create a self-seeding tile.

Select two hexagons in the seed $s'$ of the io-hTAM system to use as a self-seeding two-tile assembly.
Add two more sets of six sides for these two chosen hexagons, and replace the glue (which has strength less than $\tau$) shared by these two hexagons in the three-tile seed $s'$ with a pair of $\tau$ strength glues.
For each six-side set, the side simulating the strength-$\tau$ bond is placed adjacent to the position at angle $30^{\circ}$ (i.e. the bisector of the $[0, 60^{\circ})$ wedge containing only $IN$ sides) and has no $IN$/$OUT$ geometry.
All other sides are marked $OUT$.

Consider the behavior of the single-tile pfbTAM system: starting with a single seed tile, a second tile must come and attach via the strength-$\tau$ bond.
From then on, any tile attaching uses at least two sides to attach, all of which must be $IN$ sides, and one of which lies in a small region of the boundary containing the strength-$\tau$ bond.
Once the tile is attached, the strength-$\tau$ bond becomes unusable, as it is geometrically blocked by the neighboring tile which the $IN$ side in this region attached to.
By induction, this remains true for assemblies produced via an arbitrary number of steps. 

\fi

\section{Translation-Only Systems}

Similar to the previous section, we discuss the power and limitations of
systems composed of tiles of a single type, in terms of the aTAM tile assembly
systems which they can simulate.  However,
\ifabstract
\else
while previous assemblies were
composed of translated copies of the single tile type which were able to
represent each of multiple aTAM tile types by being positioned in specific
relative rotations to each other, here
\fi
we no longer allow rotations.
\ifabstract
\else
Instead,
while assemblies are still composed of translated copies of a single tile type,
the assembly can be thought of as lying within a logical rectangular grid, such
that each tile has a ``central body'' that is contained within the bounding
rectangle for its coordinate position, and the identity of the tile type from
the simulated aTAM system that it is representing is determined by the
position of the tile's central body relative to that bounding box.
\fi
Thus, the
simulation of multiple tile types occurs by one single tile type that can
assume a specific translation relative to a fixed coordinate system for each
tile type that it can simulate: in summary, an aTAM tile type is encoded by the relative position of our single tie type.   Therefore, in this section the term
\emph{translation} is used to describe the local translation of tiles within
the logical bounding boxes of their coordinate locations, rather than the
global translation represented by the positioning at the various coordinate
locations.

We show that such systems can simulate computationally universal systems.  This requires arbitrarily large seeds (whose size is a simple linear function of the number of timesteps in the simulated computation).  We then prove strict limitations on such systems with small seeds.

\subsection{Universality in Translation-Only Systems}
\label{sec:translation-positive}

\begin{figure}[h!]
    \begin{center}
    \ifabstract
\else
           \includegraphics[width=0.6\columnwidth]{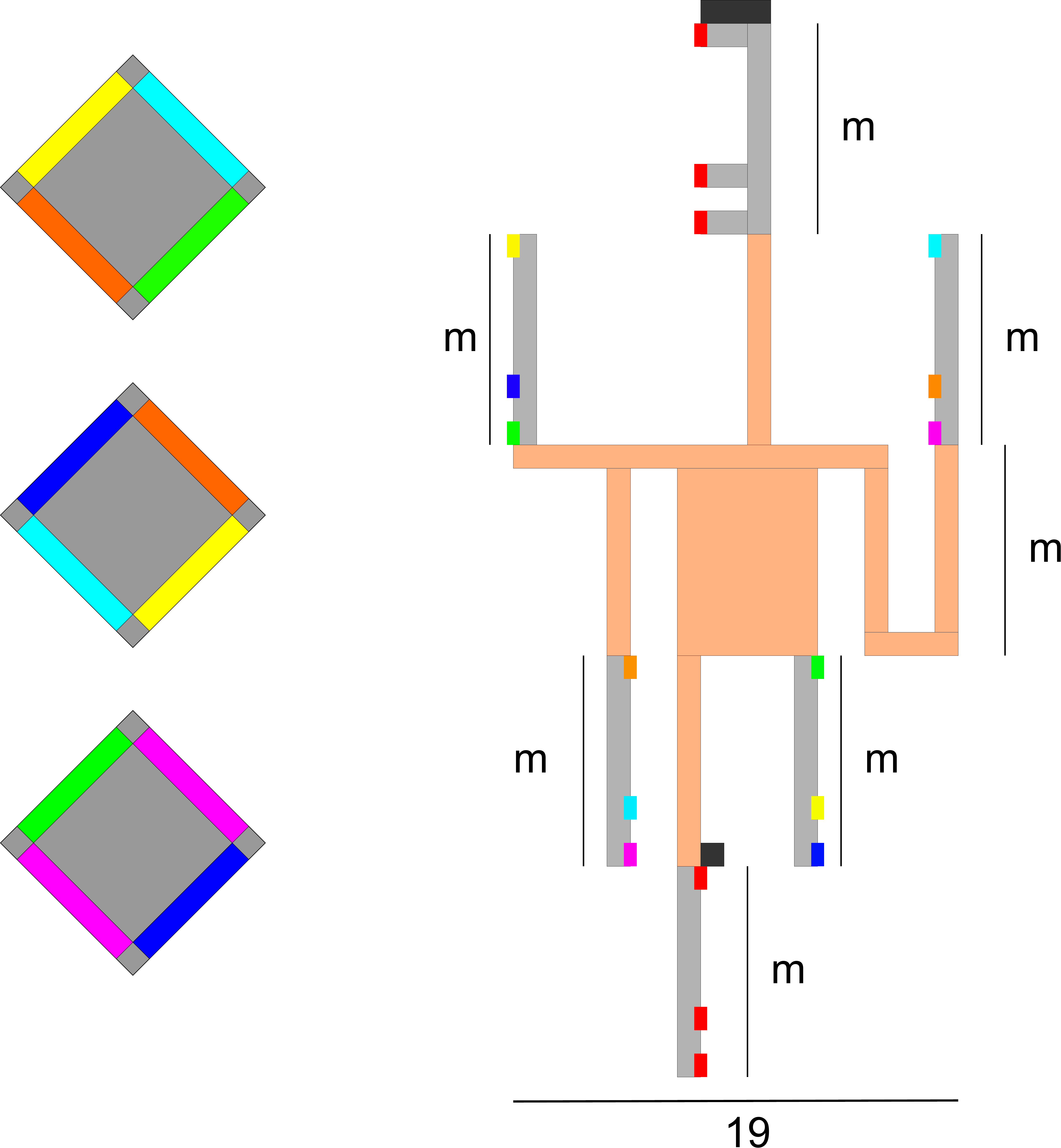}
    \fi
    \caption{
    \label{fig:sliderConversion}
A {\em slider tile} derived from a set of square tiles $T$,
\ifabstract
\else
representing each tile in $T$ by sliding up and down into different relative translations compared to its
neighbors.  The slider has 6 bonding pads, each with a sequence of glues
positioned according to a sequence of positive integers $x_1, \ldots , x_{|T|}$
that satisfy the algebraic constraints of Lemma~\ref{lemma:polynomialxi}.
The glue at position $x_i$ on the northwest pad is the glue on the northwest
face of the $i$th  tile type from $T$ (starting with $x_1$ on the southmost tip of the pad). The glues for the northeast,
southwest, and southeast pads are similarly assigned.  The top and bottom pads
also have glues assigned at positions $x_1, \ldots , x_{|T|}$, but all glue types on these
pads are of the same, neutral type.
In this figure, $m$ denotes the value
$x_{|T|}$, the largest integer in the sequence $x_1 , \ldots , x_{|T|}$.  From
Lemma~\ref{lemma:polynomialxi}, we know that an appropriate sequence exists
such that $m \leq 3|T|^5$.
\fi
}
    \end{center}
\end{figure}

\begin{figure}[t]
    \begin{center}
    \ifabstract
\else
       \includegraphics[width=0.8\columnwidth]{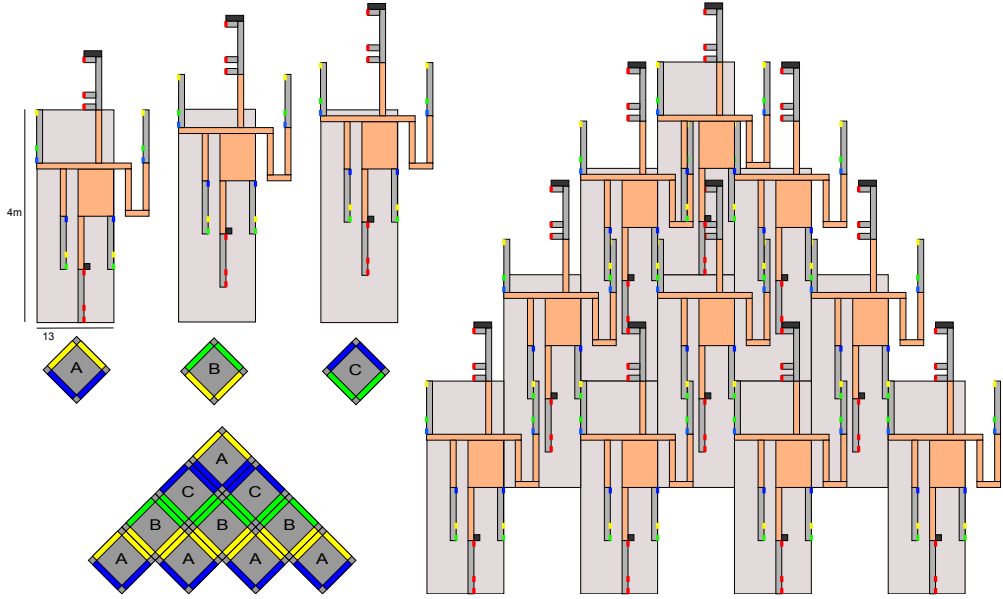}
   \fi
    \caption{
    \label{fig:sliderAssemblyMapping} An assembly of sliders for a square tile set $T$ is mapped to a corresponding square tile assembly over $T$ by placing an imaginary grid shown by the grey background blocks in this figure.
\ifabstract
\else
The tile of $T$ represented by the slider in each grey block is determined by the north/south translation of the slider in the given block.  In particular, the slider represents the $ith$ tile of $T$ if the $ith$ glue on the northwest pad of the slider lines up with the northernmost position of the grid rectangle. (Note that the assembly shown in this example does not adhere to the double-checker-boarded constraint and so such a produced assembly would not be guaranteed to form without error in the slider construction.)
\fi
}
    \end{center}
\end{figure}

\begin{figure}[t]
    \begin{center}
    \ifabstract
\else
       \includegraphics[width=0.55\columnwidth]{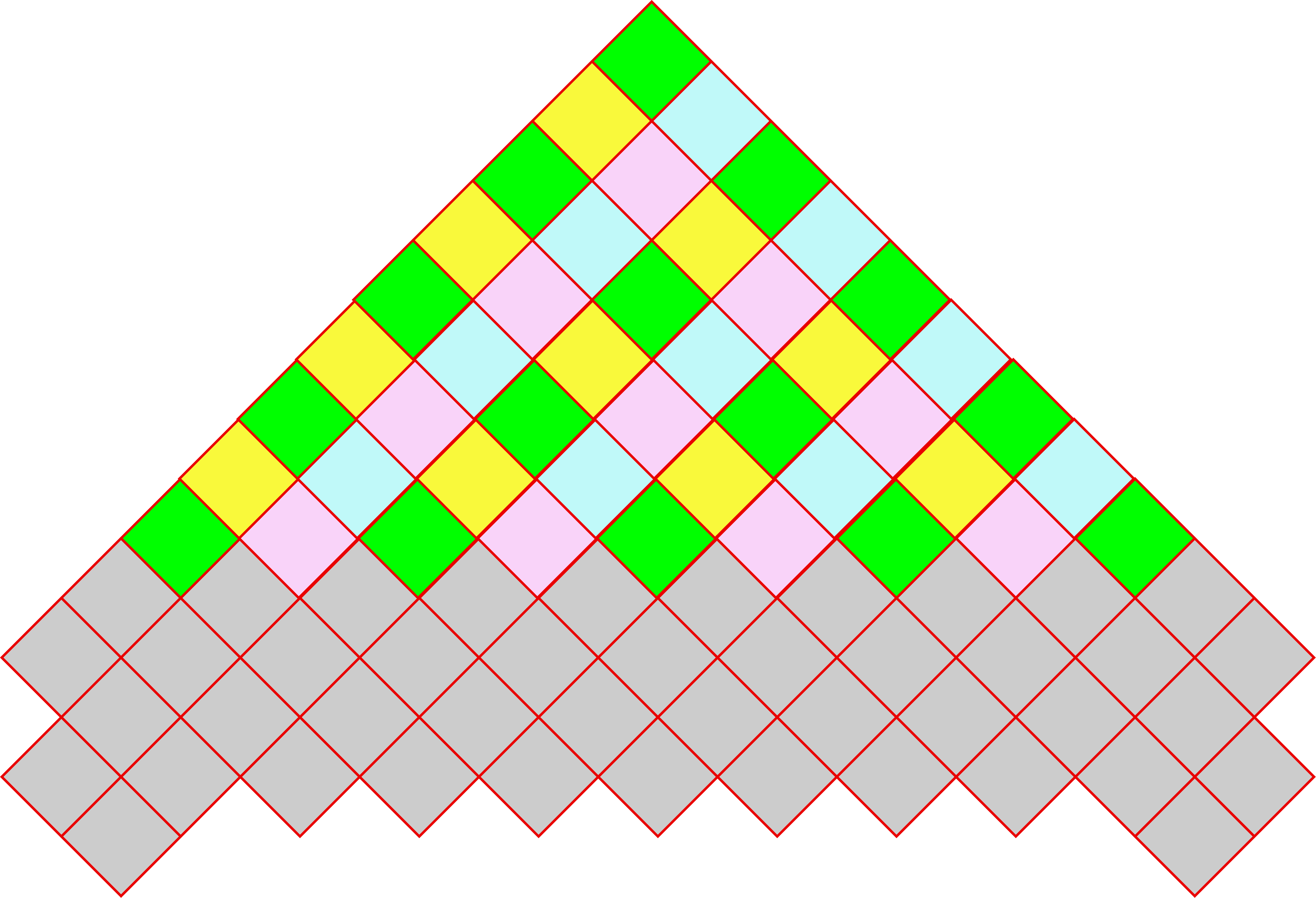}
    \fi
    \caption{Fat-seeded, double-checkerboarded pyramid systems start with a seed assembly shown in grey and grow upwards with cooperative temperature-2 bonding, yielding a maximum possible assembly in the shape of a pyramid.
\ifabstract
\else
Additionally, the system is double-checkerboarded, meaning a placed tile is distinct from its north, south, east, west, northwest, southeast, northeast, and southwest neighbors.
\fi
    \label{fig:arrowHeadColored} }
    \end{center}
\end{figure}

In this section we show that for any given 1D cellular automaton, there exists
a single tile, such that with an appropriate seed consisting of $O(n)$ copies of
this tile, the corresponding translation-only tile system will simulate the
first $n$ steps of the given cellular automaton.  The universality of 1D cellular
automata thus yields a single translation-only tile that is computationally
universal (using infinite seeds).

\subsubsection{Definitions}\label{sec:sliderDefinitions}

\ifabstract
\else

In this section, for simplicity, we will refer to unit-square aTAM tiles as sitting at a 45 degree rotation from the standard orientation.  In particular, we will refer to the four sides of such a square aTAM tile as the northwest, northeast, southeast, and southwest faces. %

\paragraph{Slider Definition}
Consider some standard square tile aTAM system $\gamma = (T, \tau, s)$.
Further, consider some total one-to-one function $X : T \rightarrow \mathbb{N}$, such
that the codomain of $X$ is a set of numbers that satisfy
Lemma~\ref{lemma:polynomialxi}.  Let this codomain be denoted by the variables
$x_1 , \ldots x_{|T|}$, such that $x_i < x_{i+1}$ for each $i$.  Let $X^{-1}$
denote the inverse of $X$.  Note that by Lemma~\ref{lemma:polynomialxi}, we
have that $1 \leq x_{|T|} \leq 3|T|^5$.  Further, without loss of
generality, assume that the leftmost, southmost tile type in $s$ is such that
$X$ maps it to $x_1$.  For such a $\gamma$ and $X$, we construct the corresponding
\emph{slider} tile, $SLDR^X_{\gamma}$,  in
Figure~\ref{fig:sliderConversion}.  The slider consists of 6 grey pads of
unit width and length equal to $x_{|T|}$.  Further, each pad has a glue
assigned to positions $x_1, \ldots x_{|T|}$ from south to north, on either the
west or east side of the grey pad as depicted in
Figure~\ref{fig:sliderConversion}.  Each of the 6 grey pads are connected by
the tan portion of the tile and the exact dimensions are provided in
Figure~\ref{fig:sliderConversion}.  In short, the width of the tile is a
constant, with the height being linear in $x_{|T|}$.  Additionally, the glue
type on the northwest pad at position $x_i$ is the glue that occurs on the
northwest face of the tile $t \in T$ such that $X(t) = x_i$.  The glue type for
the northeast, southwest, and southeast pads are defined similarly.  For the
north and south pads, the glue types at each position are all of the same glue
type $a$, where $a$ is a strength-1 glue that does not occur within the tile
set $T$.

\paragraph{Slider Assembly Mapping}
We now discuss how an assembly of slider tile types represents an assembly over
a set of square aTAM tiles $T$.  For a pictorial description of the mapping,
see Figure~\ref{fig:sliderAssemblyMapping}.  The key idea is to place an
imaginary grid of grey boxes over a given slider assembly to define the
position each slider is conceptually tiling, as well as the type of tile
represented by referencing the relative north/south translation of the slider
within the grey box.  We now formally define the mapping of a slider assembly
to a square aTAM assembly.

Consider a $\tau$-stable assembly $A$ consisting of translations of a slider tile $SLDR^X_{\gamma}$.  Now consider the westmost, southmost slider tile in $A$.  Assume this slider tile sits at coordinate position $(0, -x_1)$. We now define a partial mapping $f: Z \times Z \rightarrow Z \times Z \times T$, which maps slider coordinate locations within an assembly to both a 2D coordinate position and a tile type in $T$.

To define $f(x,y)$, first let $w = 13$ and $\ell = 4 x_{|T|}$ ($w$ and $\ell$ denote the width and height of the grey grid boxes from Figure~\ref{fig:sliderAssemblyMapping}).  Define $f(x,y)$ as follows:  If for integers $a,b \geq 0$ and $t \in \{x_1, \ldots , x_{|T|} \}$ it is the case that $x=2wa$ and $y=b\ell - t$, then $f(x,y) = (a-b,a+b, X^{-1}(t))$.  If for integers $a,b \geq 0$ and $t \in \{x_1, \ldots , x_{|T|} \}$ it is the case that $x=2wa+w$ and $y=b\ell +\ell/2 - t$, then $f(x,y) = (a-b,a+b+1, X^{-1}(t))$.  If $(x,y)$ does not satisfy either of these constraints, then $f(x,y)$ is undefined.

Given the partial mapping $f$, for a slider assembly $A$ we say $A$ maps to assembly $A'$ over $T$ if $A'$ is the assembly obtained by including each tile of type $t$ at position $(w,y)$ such that $f(x,y) = (w,u,t)$ for some slider in $A$ at position $(x,y)$.  If any slider in $A$ is at a position at which $f$ is not defined, then $A$ does not have a defined mapping to a square aTAM tile assembly over $T$.

\paragraph{Simulating an aTAM system with Sliders}
We say a slider system $\upsilon =(SLDR^X_{\gamma} , 3, s')$ \emph{terminally simulates} an aTAM square-tile system $\gamma =(T, \tau, s)$  if the set of terminal assemblies $TERM_\upsilon$ maps exactly to the set $TERM_\gamma$ when the mapping of slider assemblies is applied to each element of $TERM_\upsilon$.\footnote{This is a weaker definition of simulation than what is considered in \cite{IUSA,Versus} in that it does not model equivalent dynamics.  While our construction actually satisfies a stronger definition of simulation, we omit the more involved simulation definition for simplicity.}  For the remainder of this section we simply use the term \emph{simulates} to refer to terminal simulations.
\paragraph{Pyramid aTAM Systems}
An aTAM system $\Gamma =(T$, $2$, $s)$ is said to be a \emph{fat-seed pyramid system} if 1) $s$ contains some number $n$ of tiles configured in the format described in Figure~\ref{fig:arrowHeadColored}, with the property that all adjacent tile edges match glues, and 2) all glues in $T$ have strength 1, and 3) the tile set $T$ and seed $s$ are such that no tiles can attach to the southern face of the seed.  In addition, a fat-seed pyramid system is said to be \emph{double-checkerboarded} if for any attachable tile during the assembly process, the attached tile, its southwest, southeast, and southern neighbors are all distinct tile types.  An example of a coloring scheme that denotes which tiles must be of differing type is shown in Figure~\ref{fig:arrowHeadColored}.

\paragraph{Planar Assembly}
A pftTAM system $(T,\tau,s)$ is said to be \emph{planar} if for each possible tile attachment for all assembly sequences, the attached tile is guaranteed to have a collision-free path within the plane to slide into attachment position.  In general, planar assembly systems are desirable in that they offer the possibility for implementation within a system that is restricted to assembly on a surface, and further inform what constructions might generalize into 3 dimensions.
\fi
\subsubsection{Simulation of Cellular Automata}

In this section we show that for any double-checkerboar\-ded, fat-seed pyramid aTAM system $\Gamma$, the single-tile, trans\-lation-only slider tile system derived from $\Gamma$ simulates $\Gamma$.  It has been shown that fat-seed pyramid systems with size $n$ seeds are capable of simulating the first $O(n)$ steps of 1D blocked cellular automata~\cite{Winf98,margolus1984physics}, a universal class of cellular automata.  Thus, our result yields a single-tile translation-only system for simulating universal computation.

We first establish a lemma that makes an algebraic claim about sequences of positive integers.  This allows us to assign appropriately spaced glues to our single slider tile, such that certain undesired alignments are infeasible, and further that this can be done with a small polynomially-sized slider tile.

\begin{lemma} \label{lemma:polynomialxi}
  There is a set $\{x_1, x_2, \dots, x_k\}$ of (distinct) integers
  in the range $[1, 3 k^5]$ such that,
  for any indices $a$, $b$, $c$, $d$, $e$, and $f$,
  we have $x_a + x_b + x_c = x_d + x_e + x_f$ if and only if
  the equation holds algebraically, i.e., $\{a,b,c\} = \{d,e,f\}$.
\end{lemma}

\ifabstract
\else
\begin{proof}
  We set the $x_i$'s incrementally, mimicking a construction in the
  fusion tree data structure of Fredman and Willard \cite{Fredman-Willard-1993}.
  First we set $x_1 = 1$.
  If we have so far set $x_1, x_2, \dots, x_{i-1}$, we set $x_i$ as follows.
  For any indices $a$, $b$, $c$, $d$, $e$, $f$ $\in \{1,2,\dots,i-1\}$
  for which $x_a + x_b + x_c = x_d + x_e + x_f$ does not hold algebraically,
  the new equations involving one, two, or three copies of~$x_i$---
  \begin{eqnarray*}
  x_i + x_b + x_c &=& x_d + x_e + x_f, \\
  x_i + x_i + x_c &=& x_d + x_e + x_f, \\
  x_i + x_i + x_i &=& x_d + x_e + x_f
  \end{eqnarray*}
  ---have unique solutions for $x_i$:
  \begin{eqnarray*}
  x_i &=& - x_b - x_c + x_d + x_e + x_f, \\
  x_i &=& \textstyle {1 \over 2} ( - x_c + x_d + x_e + x_f ), \\
  x_i &=& \textstyle {1 \over 3} ( x_d + x_e + x_f ).
  \end{eqnarray*}
  Thus, if we set $x_i$ to avoid these $\leq 3 (i-1)^5$ bad values,
  then we guarantee the theorem holds on $\{x_1, x_2, \dots, x_i\}$.
  (In particular, $x_i + x_b + x_b = x_j + x_b + x_b$ will hold only if $i=j$,
   so $x_i$ is distinct from previously chosen $x_j$'s.)
  Setting $x_i$ is possible provided the number of choices for $x_i$ is
  greater than the number of bad values, i.e., $3 k^5 > 3 (i-1)^5$,
  which follows from $i \leq k$.
  Once we finally set $x_k$, we have the desired set $\{x_1, x_2, \dots, x_k\}$.
\end{proof}

\fi
We now leverage our slider construction and the above properties for the main result of this subsection.

\begin{theorem}\label{thm:sliderSimulates}
For any double-checkerboarded, fat-seeded pyramid aTAM system $\Gamma = (T, 2, s)$, there exists a single-tile, translation-only pfbTAM system that simulates $\Gamma$.  Further, the single tile of the simulating system is of size $O(|T|^5)$, and the system satisfies the planar assembly constraint.

\end{theorem}

\ifabstract
\else
\begin{proof}
Consider an arbitrary double-checkerboarded, fat-seed pyramid aTAM system
$\Gamma = (T,2,s)$ with size $n$ seed.  We will prove the theorem by showing that the single-tile translation-only system $\beta = (SLDR^X_\Gamma , 3, s')$ is a planar simulation of $\Gamma$, where $SLDR^X_\Gamma$ is a slider derived from $\Gamma$ according to Section~\ref{sec:sliderDefinitions}, and $s'$ is the assembly over $SLDR^X_\Gamma$ that maps to $s$ according to the assembly mapping described in Section~\ref{sec:sliderDefinitions}.

First, we show that
$\beta = (SLDR^X_\Gamma , 3, s')$ satisfies the requirement that $s'$ is
stable.  Note that any cut of a fat-seed must separate at least 3 neighbor
tiles, with a ``neighbor'' tile being any tile directly north, south, northeast,
southeast, northwest, or southwest.  Therefore, since all adjacent edges of $s$
are assumed to be matching strength-1 glues by the definition of fat-pyramid
systems, we know that the seed $s'$ has minimum cut strength of at least 3 (see
Figure~\ref{fig:arrowHeadColored} for a picture example of $s'$).

To finish the argument for correct simulation, we need to show that the set of producible assemblies are the same when the slider assemblies are mapped according to Section~\ref{sec:sliderDefinitions}.  To show this, assume a slider assembly $A$ correctly maps to producible pyramid assembly $A'$.  For any attachable aTAM tile to $A'$, it is easy to see that there is a corresponding position at which a slider tile can attach to $A$ to obtain a slider assembly that maps to the new aTAM assembly.  The crux of the correctness argument lies in showing that any slider that may attach must attach such that the resultant assembly is defined.  If the assembly is defined, it is straightforward to verify that the resultant assembly maps to a corresponding producible pyramid assembly.  Therefore, we will argue that the attachment of a single slider will maintain that the assembly has a defined mapping to a producible pyramid assembly.

To argue this, we first rule out a number of potential issues.  First, an
attaching slider must do so by matching each of its southwest, south, and
southeast pads with one glue.  This holds because all glues have strength 1,
and by Lemma~\ref{lemma:polynomialxi} the sequence of
$x_i$'s are such that at most one pair per pad can line up, assuming a
non-perfect alignment.  In the case of perfect alignment of a southwest or
southeast pad, there will be 0-strength bonding because of the double-checkerboarded
structure of the simulated tile system.  Finally, perfect alignment of the
southern pad is prevented by geometric hindrance in the form of the black bump
protrusions at the base of the southern and northern pads.

Therefore, the only way to extend an existing arrangement of tiles is to
attach a new tile to three different tiles, making use of the southeast, south, and southeast pads,
each with glue strength one. We will proceed by induction to argue that such a newly placed
tile will have to be apropriately placed on the underlying grid, with a vertical shift corresponding
to the appropriate tile type of the simulated tile system.

\begin{figure}[t]
    \begin{center}
    \ifabstract
           \includegraphics[width=0.64\columnwidth]{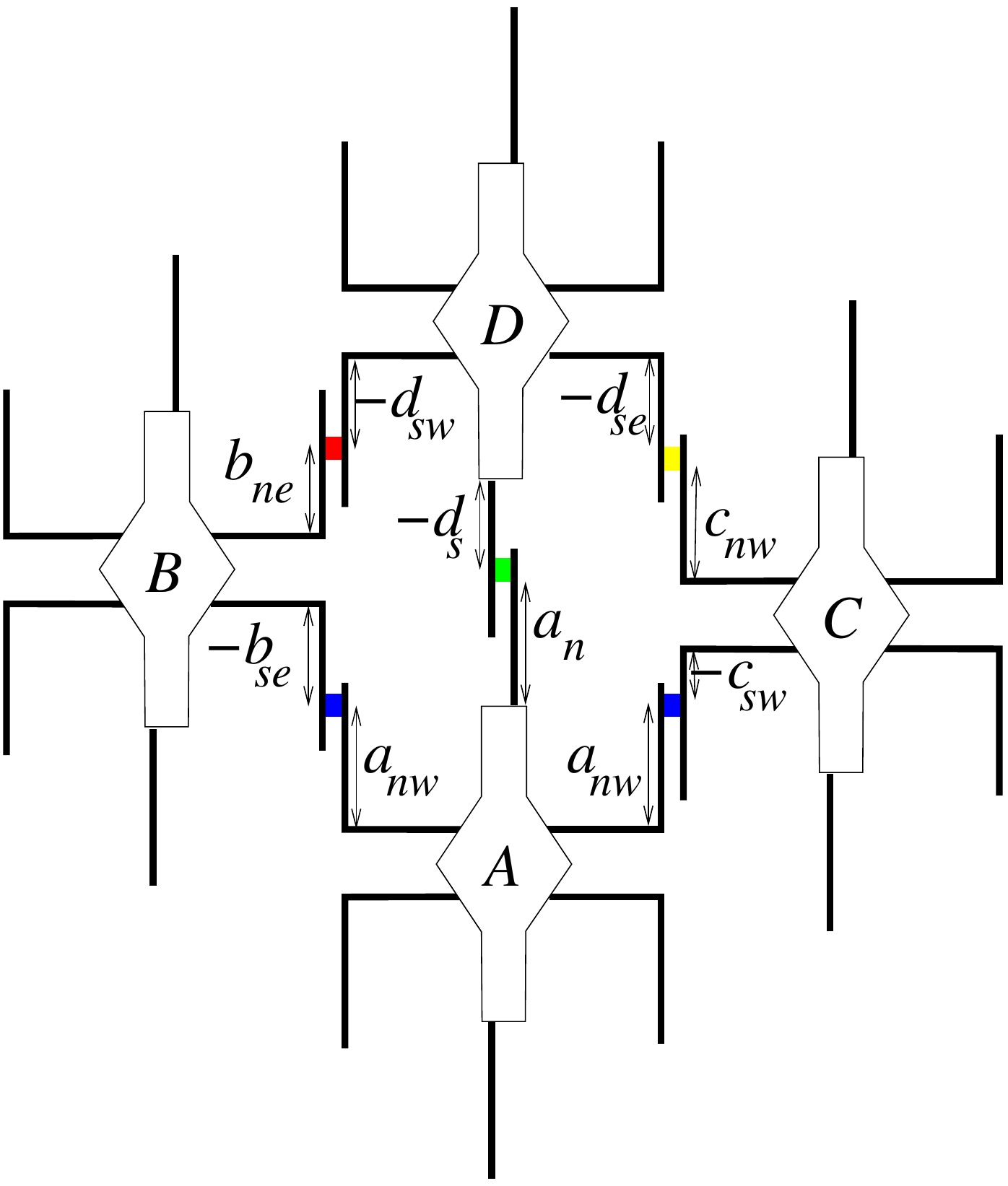}
    \else
           \includegraphics[width=0.6\columnwidth]{images/13}
    \fi
    \caption{Correctness of the slider simulation: A new slider $D$ is attached to existing sliders, $A$, $B$, $C$, making use
of three different
\ifabstract
bonds.
\else
bonds, shown in red, green, yellow.
\fi
Arguing the relationship between the involved $y$-distances
\ifabstract
\else
(which are measured from the baseline of each pad, and considered modulo $m$)
\fi
\ifabstract
shows that this is possible iff a position for $D$ encodes a tile in the simulated tile system.
\else
shows that this is possible iff a position for $D$ is used that corresponds to an encoding of the corresponding tile in the simulated tile system.
\fi
    \label{fig:tripleThreat} }
    \end{center}
\ifabstract
\vspace*{-6mm}
\fi
\end{figure}

More precisely, consider a tile $D$ that is bonded with the existing tiles $A$, $B$, $C$, as indicated in
Figure~\ref{fig:tripleThreat}. In the following, we discuss the $y$-distance of the involved slider
pegs from the baseline of their respective pads, which causes a vertical deviation from the respective
grid positions. In the following, we consider these deviations modulo $m$, and simply refer to these as ``relative positions''.
Reference to slider name is indicated by letters $a$, $b$, $c$, $d$,
while indices $sw$, $s$, $se$, $ne$, $n$, and $nw$ encode the pegs in directions southwest, south, southeast, northeast, north, northwest; for example, the
bond between $A$ and $D$ (shown in green in the figure)
is shifted by a (vertical) $y$-distance of $d_{s}$ from the baseline of $D$'s northern bonding pad,
and $a_n$ from the baseline of $A$'s southern bonding pad.
By assumption, $A$, $B$, $C$ are correctly placed, so that
$A$ bonds with $B$ and $C$ (indicated by blue color in the figure)
at the peg distance that encodes the tile type of $A$, i.e., $a_{nw}=a_{ne}$.
We denote by $\Delta_y(X,Z)$ the difference of relative positions between two
tiles $X$ and $Z$.
Then we have $\Delta_y(B,A)=a_{nw}-b_{se}$ and $\Delta_y(C,A)=a_{ne}-c_{sw}$.

Now assume that there is a bond between $D$ and $B$ (shown in red) that uses pegs at relative positions $d_{sw}$ and $a_{ne}$,
a bond between $D$ and $A$ (shown in green) that uses pegs at relative positions $d_{s}$ and $a_{n}$, and
a bond between $D$ and $C$ (shown in yellow) that uses pegs at relative positions $d_{se}$ and $c_{nw}$.
We will argue that this implies that $D$ is correctly placed, with all bonds of any tile using a proper peg position,
thus $D$ and its bonds encode a tile of appropriate tile type in the original tile system.

First of all, we observe that $\Delta(D,B)=b_{ne}-d_{sw}$, while $\Delta(D,C)=c_{nw}-d_{se}$, and $\Delta_x(D,A)=a_{n}-d_{s}$.
Comparing the total shift between the three paths $(A,B,D)$, $(A,C,D)$, $(A,D)$, as shown in the figure, we conclude
$\Delta_x(D,B)+\Delta_x(B,A)=\Delta_x(D,C)+\Delta_x(C,A)$, as well as
$\Delta_x(D,C)+\Delta_x(C,A)=\Delta_x(D,A)$. From the first equation, we conclude
$b_{ne}-d_{sw}+a_{nw}-b_{se}=c_{nw}-d_{se}+a_{ne}-c_{sw}$, so $a_{nw}=a_{ne}$ implies
\begin{equation}
b_{ne}+c_{sw}+d_{se}=b_{se}+c_{nw}+d_{sw}.
\label{eq:abd=acd}
\end{equation}
From the second equation, we conclude
\begin{equation}
a_{nw}+b_{ne}+d_{s}=a_{n}+b_{se}+d_{sw}.
\label{eq:abd=ad}
\end{equation}

Now both $b_{ne}=d_{sw}$ and $d_{se}=c_{nw}$ can be excluded, as they would imply a perfect alignment between
$B$ and $D$, or $D$ and $C$, respectively.

Next, consider $b_{ne}\neq b_{se}$. Then Lemma~\ref{lemma:polynomialxi} applied to \eqref{eq:abd=acd}
and \eqref{eq:abd=ad} implies that $d_{se}=b_{se}$, $c_{sw}=d_{sw}$,
$b_{ne}=c_{nw}$, $a_{nw}=d_{sw}$, $b_{ne}=a_n$, $d_s=b_{se}$. This implies $c_{sw}=d_{sw}=a_{nw}$; because of
$a_{nw}=a{_ne}$, it follows that $a_{ne}=c_{sw}$, i.e, perfect alignment between $A$ and $C$, which is impossible.

Therefore, we conclude that $b_{ne}=b_{se}$; then Lemma~\ref{lemma:polynomialxi}
applied to \eqref{eq:abd=acd}
implies that $c_{sw}=c_{nw}$, $d_{se}=d_{sw}$, and applied to \eqref{eq:abd=ad}
implies that $a_{sw}=a_{n}$, $d_{s}=d_{sw}$, meaning that all bonds of $D$ must use proper pegs for
encoding the simulated tiles, as claimed.
\end{proof}
\fi

\subsection{Limitations of Translation-Only Systems}
\label{sec:inmemoryofporky}
The many-gon single-tile self-assembly systems derive their power from being able
to rotate.
\ifabstract
\else
It is natural to ask whether such rotations are
necessary for general single-tile simulations, or if it is possible to
simulate different tile types by the attachment of a single geometrically
complex tile at different relative translations, in an even more general fashion than Subsection~\ref{sec:translation-positive}.
\fi
\ifabstract
\else
We show that rotations are
necessary in the single-tile model, by proving that translation-only systems with a single tile
have very limited power, regardless of the tile's geometric complexity.
This holds under rather general assumptions: for the purposes of this section,
a tile may be an arbitrary two-dimensional, bounded, connected, regular closed set $S$,
i.e., $S$ is equal to the topological closure of its interior points.
In the following, we say that two tiles {\em overlap}, iff they have non-disjoint
interiors; they {\em touch}, iff they intersect without overlapping.
A potential bond between two tiles requires that they touch in more than
one point, which must be equipped with a matching glue.
\fi

\ifabstract

\else
To obtain the impossibility results, we
start with a lemma about the
translation of connected shapes in 2D.  We assume that this lemma
has been previously discovered, but we have been unable to find it in
the literature.
\fi

\begin{lemma}
\label{le:chain}
Consider a two-dimensional, bounded, connected, regular closed set $S$,
i.e., $S$ is equal to the topological closure of its interior points.
Suppose $S$ is translated by a
vector $v$ to obtain shape $S_{v}$, such that $S$ and $S_{v}$ do not overlap.
Then the shape $S_{c*v}$ obtained by translating $S$ by $c*v$ for any
integer $c \neq 0$ also does not overlap $S$.  \end{lemma}

\ifabstract
\else
\begin{proof}
Assume that there is a smallest integer $c>1$ for which $S$ and $S_{c*v}$ overlap; we will show
that $S_v$ must overlap one of them, implying the claim.
Without loss of generality, let $v=(1,0)$; see Figure~\ref{fig:overlap}.

\begin{figure}[htp]
\centering
        \includegraphics[width=.6\columnwidth]{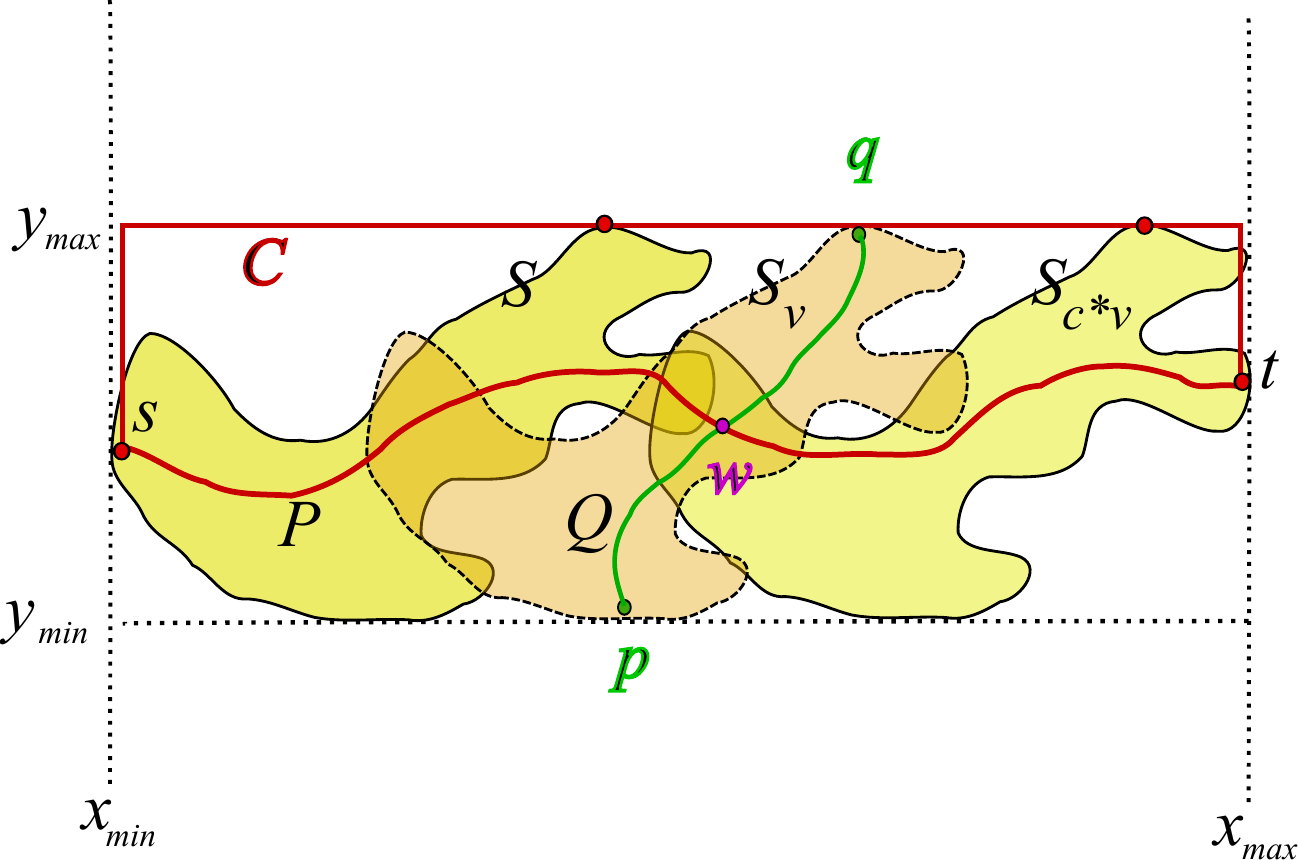}
  \caption{Proof of Lemma~\ref{le:chain}: overlaps between multiple copies of a shape. }
  \label{fig:overlap}
  \centering
\ifabstract
\vspace*{-3mm}
\fi
\end{figure}

Let $x_{\min}$ be the smallest $x$-coordinate of $S$, and let $x_{\max}$ be the largest
$x$-coordinate of $S_{c*v}$. For a small $\varepsilon>0$,
let $s=(s_x,s_y)$ be an interior point of $S$ with $s_x-x_{\min}<\varepsilon$, and let
let $t=(t_x,t_y)$ be an interior point of $S_{c*v}$ with $x_{\max}-t_x<\varepsilon$.
Because $S$ and $S_{c*v}$ are regular closed, and $S\cup S_{c*v}$ overlap,
there is a path $P$ between $s$ and $t$ that stays in the interior of
$S\cup S_{c*v}$. Let $y_{\max}$ and $y_{\min}$ be the largest and smallest
$y$-coordinate of points in $S$. Then we can connect $t$ with $s$ by a vertical
line segment $\ell_1$ up to $y$-coordinate $y_S$, a horizontal line segment $\ell_2$ to
the $x$-coordinate of $s$, and a vertical line segment $\ell_3$ to $s$; this yields
a simple closed curve $C$.

Because any point of $S_{v}$ with maximal $y$-coordinate
lies on $\ell_2$, and $P$ lies strictly below $\ell_2$,
$C$ must contain an interior point $q$ of $S_{v}$ in its interior.
On the other hand, a point of $S_{v}$ with $y$-coordinate $y_{\min}$
must lie below $P$, and therefore outside of $C$, so there must be an interior point $p$ of $S_{v}$
outside of $C$. Because $S_{v}$ is connected, there must be a path $Q$ in the interior of
$S_{v}$ that connects $p$ and $q$. Thus, $Q$ must cross $C$, hence $P$. Therefore,
there is a point $w$ in which $S_{v}$ and
$S\cup S_{c*v}$ overlap, and the claim follows.
\end{proof}
\fi

\begin{theorem}
\label{th:onechain}
For any self-seeding, single-tile, transla\-tion-only (non-rotatable, non-flippable) self-assembly system $\Gamma
= (T,\tau)$, the set of producible assemblies of $\Gamma$ is either just the
single seed copy of $T$, or contains assemblies of unbounded size.
\end{theorem}

\ifabstract
\else
\begin{proof}
Consider a seed tile $T_0$ and suppose that we can attach a
translated copy $T_1$ to $T_0$ without causing any overlap. Then Lemma~\ref{le:chain}
shows that we can proceed to assemble an unbounded sequence of tile copies $T_i$,
by attaching each $T_i$ to $T_{i-1}$.
\end{proof}
\fi

\begin{corollary}\label{cor:trans-cant-sim-all-atam}
There are aTAM systems that cannot be simulated by a self-seeding, translation-only, 1-tile self-assembly system.
\end{corollary}

\ifabstract
\else
Corollary~\ref{cor:trans-cant-sim-all-atam} clearly follows from
Theorem~\ref{th:onechain}, because there exist singly-seeded aTAM systems
(those with seeds consisting of a single tile) with terminal assemblies that
have more than one, but a finite number of tiles.
\fi

\ifabstract
\else
\begin{figure}[htp]
\centering
        \includegraphics[width=.3\columnwidth]{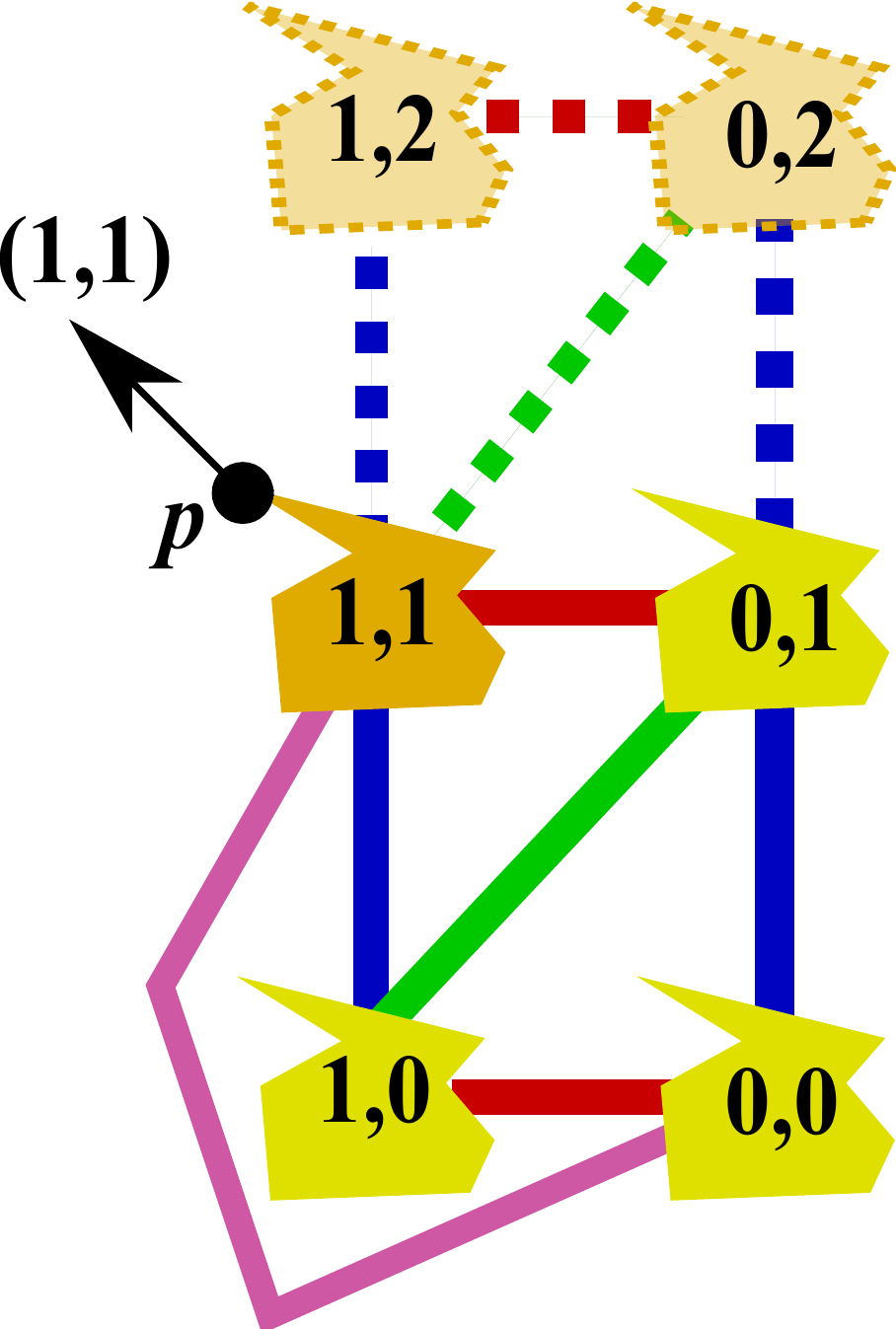}
  \caption{Proof of Theorem~\ref{th:porkysdeath}: establishing an unbounded assembly. }
  \label{fig:porkysdeath}
  \centering
\end{figure}
\fi

Theorem~\ref{th:onechain} shows an inherent problem with self-seeding
trans\-lation-only systems:  to get started requires a strength-$\tau$
attachment between two individual tiles, which leads to infinite growth.
This remains true in the more general case of weaker individual
bonds, where a
tile requires a cooperative bonding between two or more copies copies of itself
to attach in a $\tau$-stable fashion.  In particular, for any translation-only
single-tile system whose seed is a 3-tile assembly, any additional attachment
of even a single tile implies the assembly can grow forever, showing that
translation-only single-tile systems are much weaker than rotational systems.

\begin{theorem}
\label{th:porkysdeath}
Let $\Gamma = (T,\tau,S)$ be a self-assembly system with $|T|=1$, consisting of
a single non-rotatable, non-flippable tile $T_0$ that is closed regular.
If $\Gamma$'s seed $S$ consists of three copies of $T_0$ that are bonded
in a $\tau$-stable manner, the set of producible assemblies of $\Gamma$ is
either just this seed, or contains assemblies of unbounded size.
\end{theorem}

\ifabstract
\else
\begin{proof}
Refer to Figure~\ref{fig:porkysdeath}, in which $T_0$ is shown
symbolically by a polygon; note that $T_0$'s shape may be
much more complicated, and that there may be various intersection
points between different copies. A tile is in a geometrically feasible
position if it does not overlap any existing tile; it can attach
to an existing assembly, iff in addition, there are possible bonds of sufficient
strength.

We start by considering a seed consisting of
three mutually touching copies $T_1$, $T_2$, $T_3$ of $T_0$,
denoted by the labels $(0,0)$, $(1,0)$, and $(0,1)$ in the figure.
Let the respective bonds be denoted by ``blue'', ``red'',
and ``green'', with strengths $\sigma_b$, $\sigma_r$, and
$\sigma_g$.   If
the centers of the three tile copies form a collinear arrangement,
or if one of the individual bonds has strength at least $\tau$, the claim
follows immediately from Lemma~\ref{le:chain}, along
the lines of Theorem~\ref{th:onechain}.
Therefore, we assume that
$\sigma_b$, $\sigma_r$, $\sigma_g<\tau$, but
$\sigma_b+\sigma_r\geq\tau$,
$\sigma_b+\sigma_g\geq\tau$,
$\sigma_r+\sigma_g\geq\tau$; furthermore, we may assume
that the tile centers form a non-degenerate triangle,
which spans a two-dimensional vector space
described by the basis vectors $(1,0)$ (between $T_1$ and
$T_2$) and $(0,1)$ (between $T_1$ and $T_3$.)
In the following, we will denote tile copies by their
respective coordinates, i.e.,
$T_1=T(0,0)$, $T_2=T(1,0)$, $T_3=T(0,1)$.

Now consider a potential tile $T(1,1)$.
Because its position relative to $T(0,1)$ is the same
as that of $T(1,0)$ relative to $T(0,0)$, $T(1,1)$ and $T(0,1)$
touch, but their interiors do not overlap, and
they can form a red bond. Similarly,
$T(1,1)$'s position relative to $T(1,0)$ is the same as that
of $T(0,1)$ relative to $T(0,0)$, so $T(1,1)$ and $T(1,0)$ may form
a blue blond. Because of $\sigma_b+\sigma_r\geq\tau$,
$T(1,1)$ can attach to the seed, provided it is
in a geometrically feasible position; because $T(1,1)$
does not overlap $T(0,1)$ or $T(1,0)$, this is the case,
iff its interior does not intersect the interior of $T(0,0)$.

So assume that $T(1,1)$ and $T(0,0)$ overlap.
(In Figure~\ref{fig:porkysdeath}, this is shown symbolically by a purple connection.)
Then the intersection
graph of $T(0,0)$, $T(1,0)$, $T(0,1)$, $T(1,1)$ is a planar geometric
embedding of the complete graph $K_4$; as bonds require more than a single
intersection point, we can represent the edge $e_{i,j}$
between any pair of
vertices $v_i$ and $v_j$ by a path strictly within the union of the respective
tiles, without intersecting any of the two other tiles.
Because we have a planar embedding of $K_4$,
this means that one of the tiles must be strictly
inside of a simple closed curve that
lies completely within the union of the other three.
This leads to a contradiction:
an extreme point $p$ of $T(1,1)$ in direction $(1,1)$ cannot be
in the convex hull of $T(0,0)\cup T(1,0)\cup T(0,1)$. Similarly,
each other tile has an extreme point outside of the convex hull
of the three others. Therefore, the intersection graph
of the tiles cannot be a $K_4$, showing that $T(1,1)$ and $T(0,0)$
cannot overlap.

We conclude that $T(1,1)$ is a feasible extension of the
seed. By Lemma~\ref{le:chain},
$T(0,0)$ and $T(0,2)$ cannot overlap; furthermore, we can apply
the above reasoning for the non-overlapping position of $T(1,1)$
and $T(0,0)$ to show that $T(1,0)$, $T(0,1)$, $T(1,1)$, $T(0,2)$
cannot form a $K_4$, so $T(0,2)$ and $T(1,0)$ do
not overlap.  Hence, we conclude that $T(0,2)$ is a further
feasible extension, using green and blue bonds of combined strength
$\sigma_r+\sigma_g\geq\tau$.

This can be iterated: by alternating extensions at
$T(1,i)$ and $T(0,i)$, we get an unbounded assembly.
\end{proof}
\fi

Matters get more
involved with arbitrary seeds. The following conjecture
implies that
translation-only, single-tile systems have significantly
reduced computing power.

\begin{conjecture}
\label{con:generalseed}
Let $\Gamma = (T,\tau)$ be a self-assembly system with $|T|=1$,
consisting of a single non-rotatable, non-flippable
tile $P$ that is closed regular, i.e.,
a closed two-dimensional set that is equal to the closed hull
of its interior points. If $\Gamma$ is self-seeded with
a finite number $n$ of copies of $P$ that are bonded in a $\tau$-stable manner,
then (regardless of the geometric complexity of $P$),
any producible assembly of $\Gamma$ consists either of $O(n^2)$ tiles,
or is unbounded. Furthermore, the diameter of any finite assembly will
be linear in the diameter of the seed.
\end{conjecture}

A proof should be based on Lemma~\ref{le:chain}, Theorem~\ref{th:porkysdeath},
and exploit the planarity of the tile-touching graph during the assembly,
in particular the length of its exterior face.

Our slider construction yields an assembly with $\Omega(n^2)$ tiles
for a seed with $n$ tiles, so the bound would be tight.

\section{Plane Tilings}
\label{Plane Tilings}

Here we apply the idea of bump-and-dent geometry in the $n$-gon polygon tile assembly construction to another, substantially older type of ``tile system''.
Plane tiling has been studied for many decades, and the standard problem in this area is the following:  given a set of polygon-shaped tiles with patterns on them, find an infinite arrangement of these tiles that covers the plane, subject to satisfying constraints on the patterns.
Wang tiles are one such family of plane tiling systems: a Wang tile consists of a square tile with each edge colored, with the constraint that each edge must be matched to an adjacent edge of another tile with the same color.
However, there are a large number of other plane tiling systems such as Penrose tilings, Robinson tilings, and the interlocking tessellations of M.C. Escher.

The complexity of such systems comes from extending both the set of transformations the tiles may undergo and the rules enforced on adjacent tiles.
With the notation of Section~\ref{def:planetiling}, classical square Wang tiles belong to the plane tiling family $(\{\}, c_m)$, and so may not be rotated or flipped; adjacent tiles must have identical colors on coincident tile sides.
Square Robinson tiles~\cite{Robinson-1971}, on the other hand (belonging to $(\{t_r, t_f\}, c_c)$), may be rotated or flipped, but  have complementary color patterns on coincident tile sides\ifabstract\else  ~(see Figure~\ref{fig:aperiodic_ten_tiles}). \fi

Clearly, Wang tile systems have a strong resemblance to the square-tile aTAM systems discussed in this paper. However, the models are qualitatively different: the behavior of a set of Wang tiles, and a similar aTAM system with edge colors exchanged for glues, are unrelated.
For instance, coincident edges of adjacent tiles in an {\em assembly} may have different glues while coincident edges in a {\em plane tiling} are required to have matching colors.
Also, partially completed tile assemblies are guaranteed to be contained in some terminal assembly (and are computable in polynomial time on a Turing machine), while partial tilings are {\em not} guaranteed to be contained in any plane tiling. 

Nevertheless, the modest similarities are sufficient to adapt the bump-and-dent approach from Section~\ref{sec:manygons} to create a canonical set of nearly-plane tiling systems that simulate any plane tiling system of colored squares or hexagons. 

\paragraph{Our results}   
\ifabstract
\else
We consider the set of plane tiling systems covered under the definition found in Section~\ref{def:planetiling}, restricted to squares and hexagons with colored patterns and/or small surface geometries along their edges.  
Essentially, we provide a straightforward method to convert each member of this wide class of plane tiling systems into a single-tile system on the same lattice. 
Specifically, we show how any such tile system in the families $(\{\}, c_m)$, $(\{t_r, t_f\}, c_m)$, $(\{\}, c_c)$, and $(\{t_r, t_f\}, c_c)$, can be simulated by a single-tile nearly-plane tiling system in the family $(\{t_r, t_f\}, c_c)$.
\begin{theorem}
Each colored square and hexagon plane tiling system in the families 
$(\{\}, c_m)$, $(\{t_r, t_f\}$, $c_m)$,   $(\{\}, c_c)$ and $(\{t_r, t_f\}$, $c_c)$ 
is simulated by an $n$-gon nearly-plane tiling system.
\end{theorem} 

We show that as a corollary there exist single-tile plane tiling systems that are aperiodic. We also get that there are single-tile plane tiling systems that are intrinsically universal: they simulate {\em all systems} with permitted sets of transformations or tile-adjacency constraints.

\paragraph{An $n$-gon nearly-plane tiling system}
Our simulator single-tile plane tiling systems borrow the idea of bump-and-dent geometry found in the self-assembly construction in Section~\ref{sec:manygons}, but modify its usage to fit the setting of plane tilings (as opposed to a  self-assembly-based pfbTAM model).
This resulting tiling system consists of a single polygon that is a convex regular polygon with small geometry added to each side.
We add the special constraint that any valid tiling using this system must consist of tiles placed on a square or hexagonal lattice, with tiles adjacent on the lattice meeting at a pair of sides that have matching geometry (see Section~\ref{sec:models} for definitions) and matching colors.
Note that this does not completely cover the plane (hence the qualifier ``nearly''), but does produce a dense infinite pattern.

\subsection{Simulating plane tiling systems}

Recall that in Section~\ref{sec:manygons}, unwanted rotations of tiles are eliminated by replacing glues on each pair of opposing sides with colors (glues) unique to that direction pair (via a numbering scheme for creating glue subsets).
The remaining unwanted rotations are eliminated via geometry, using a bump-and-dent pair located on the clockwise or counterclockwise end of each edge.
The result was that a hTAM system (similar to a hexagonal plane tiling system of the family $(\{\}, c_m)$) was converted to a single-tile pfbTAM system (similar to a single-tile plan tiling system of the family $(\{t_r, t_f\}, c_m)$).
Simulating systems in the family $(\{t_r, t_f\}, c_m)$ is possible by eliminating all direction-specific glues and surface geometry.

\begin{lemma}
All colored square and hexagonal plane tiling systems in the families $(\{\}, c_m)$ and $(\{t_r, t_f\}$, $c_m)$ can be simulated by an $n$-gon nearly-plane tiling system.
\end{lemma}

\begin{proof}
Restricting the $n$-gons to live on lattice of $S$ and adjoin only other $n$-gons with matching geometry and colors (glues in the pfbTAM construction) yields a construction simulating systems with $T = \{\}$.
Eliminating direction-specific glues and surface geometry leaves the $n$-gon tile free to rotate and reflect, thus simulating systems with $T = \{t_r, t_m\}$ directly.
\end{proof}

Next we consider simulating systems with complementary tile adjacency constraints (i.e. $C = c_c$), which we first show is reducible to a complementary geometry constraint.
Given some surface geometry, convert the geometry into a pattern by mapping the height of the geometry relative to the flat geometry to a color value (dents have negative values).
The opposing mapping is also possible: for a given pattern consisting of a function from distance along the tile's side to a color, map each color to a height of bump.
Division of the colors into primal-dual sets of complementarity yields a partition of the colors into bump colors and dent colors.

Given a translation-only geometry-complementation system, create a distinct color for each opposite-direction pair, use the clockwise-counterclockwise placement of bumps-and-dents as in the construction in Section~\ref{sec:manygons}, and place a small copy of the geometry from the side of each tile in a small interval along the corresponding side of the $n$-gon tile.
Given a system in the family $(\{t_r, t_f\}, c_c)$, create a similar construction but with a single `blank' color on all sides of the $n$-gon.

\begin{lemma}
All colored square and hexagon plane tiling systems in the families $(\{\}, c_c)$ and $(\{t_r, t_f\}$, $c_c)$ can be simulated by an $n$-gon nearly-plane tiling system.
\end{lemma}

\begin{proof}
Simulating systems from the family $(\{\}, c_c)$ using the described construction does not allow any rotations except those corresponding to an unrotated and unreflected tile in the original system, as the distinct glue pairs and geometry forbid unwanted orientations for the same reasons as in Section~\ref{sec:manygons}.
Systems from the family $(\{t_r, t_f\}, c_c)$ are also simulated correctly, as any rotation or reflection of the $n$-gon corresponds to an orientation of some tile in the simulated set, and has the same geometry as the simulated tile.
\end{proof}

\subsection{An example simulation}

\paragraph{A small aperiodic set of square tiles}

First we describe the ten-tile square plane tile set from the family $(\{t_r, t_f\}, c_c)$ of Robinson~\cite{Robinson71} that yields only aperiodic tilings of the plane.
This system uses a pattern arrows on each tile, with the constraint that every arrow head must meet an arrow tail on an adjacent tile and vice versa: every arrow tail must meet an arrow head.
Figure~\ref{fig:aperiodic_ten_tiles} shows the tile set itself with the aperiodic pattern (orange and blue) and the parity-enforcing pattern (grey) and Figure~\ref{fig:aperiodic_tiling_example} shows an example tiling using this tile set.

\begin{figure}[htp]
    \begin{center}
    \includegraphics[width=.6\columnwidth]{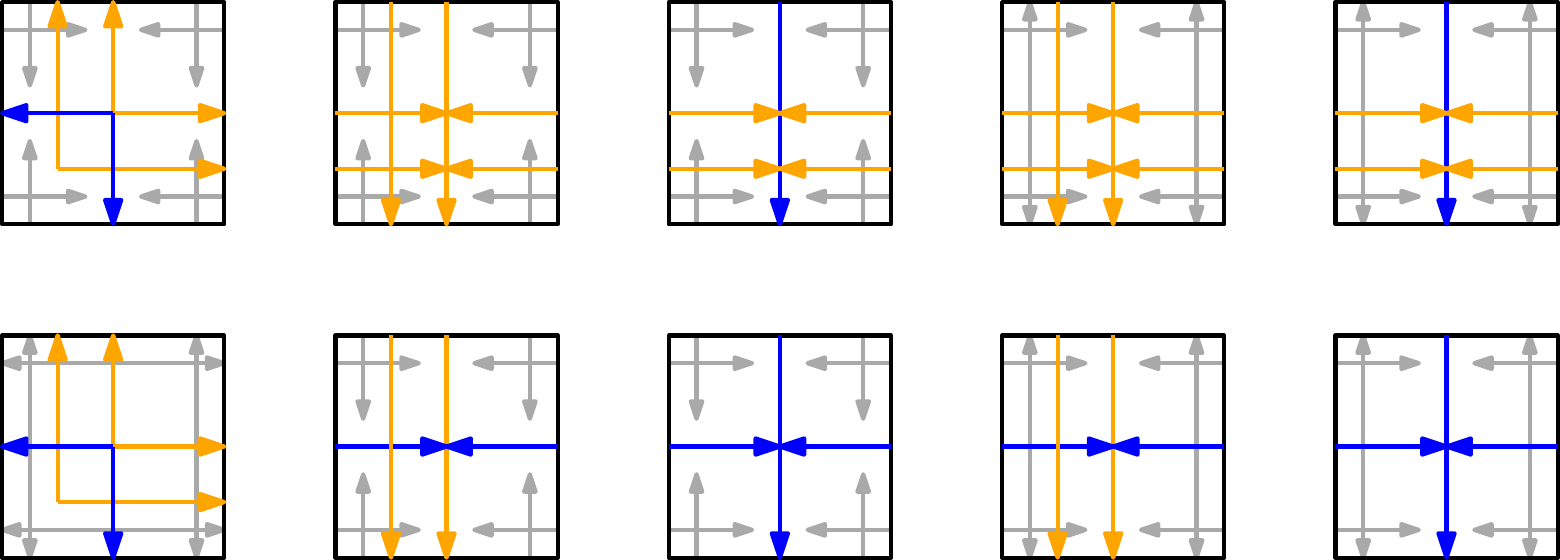}
    \caption	{
    \label{fig:aperiodic_ten_tiles}
    A set of ten aperiodic tiles from Robinson.
    The pattern constraint requires every arrow head on the surface of a tile to be met by a parallel arrow tail on an adjacent tile and vice versa.}
    \end{center}
\end{figure}

\begin{figure}[htp]
    \begin{center}
    \includegraphics[width=.6\columnwidth]{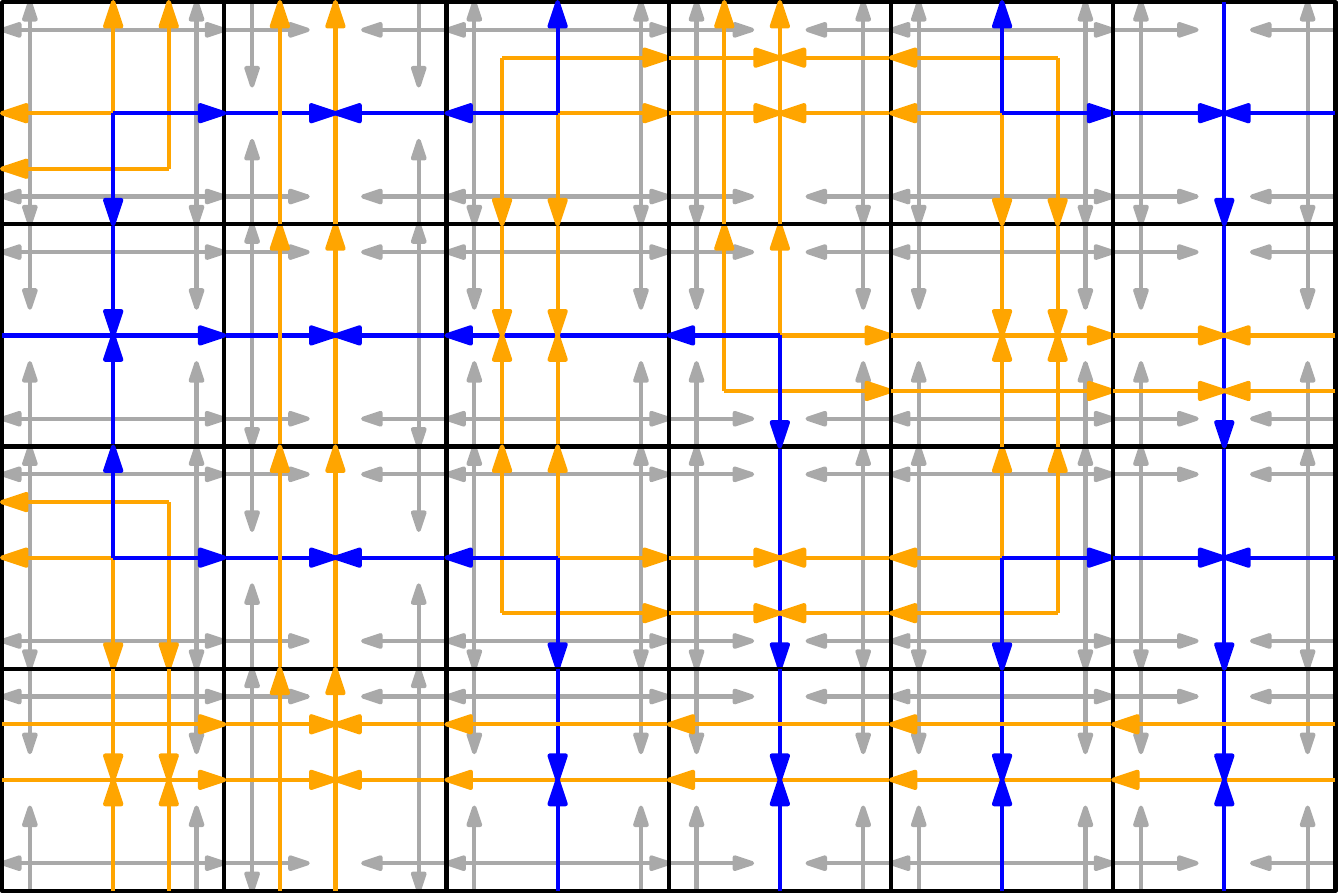}
    \caption{
    \label{fig:aperiodic_tiling_example}
    A (partial) tiling of the plane using the tile set from Robinson.
    }
    \end{center}
\end{figure}

The blue and orange arrows represent constraints in the five `basic tiles' of Robinson Figure 2.
The grey arrows represent parity constraints introduced in Robinson Figure 3 to eliminate the corner geometry of the six (l-most) square tile set found in Robinson Figure 1.

\paragraph{The simulating $n$-gon system}

We use the construction described previously to convert this ten-tile system into an $n$-gon nearly-plane tiling system.
Figure~\ref{fig:aperiodic_ten_tile_manygon} shows the resulting $4c$-gon tile.

\begin{figure}[htp]
    \begin{center}
    \includegraphics[width=.6\columnwidth]{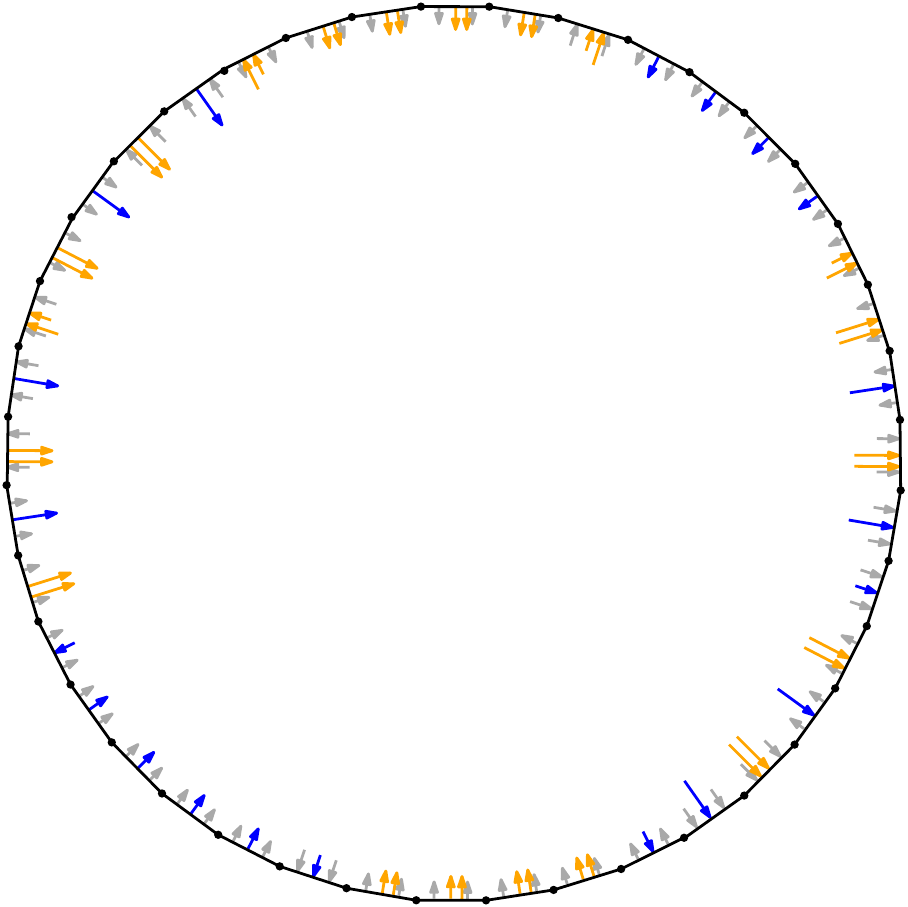}
    \caption{
    \label{fig:aperiodic_ten_tile_manygon}
     The single tile type generated from Robinson's set of ten square tiles.
    The sides of the tile correspond to the sides of every tile in the Robinson tile set such that if the tile's uppermost side is aligned horizontally then the four horizontal and vertical edges correspond to an orientation of a tile in the Robinson set.}
    \end{center}
\end{figure}

\subsection{One tile that simulates Wang tilings}
The previous construction showed that it is possible to take Robinson's plane tiling system, and convert that to a nearly-plane tiling system with a single $n$-gon tile. We can apply this transformation idea to {\em any} plane tiling system on the square or hexagonal grid in order to get a single tile, nearly-plane tiling system. As an immediate corollary, the single tile simulator inherits all (non-trivial) properties of the system being simulated. For example, we know from the previous section that there is a single tile system that tiles the plane aperiodically. 

Lafitte and Weiss introduce simulations between tilings~\cite{lafitte2007universal}. They show the existence of universal tilings: such a tiling is capable of simulating a countably infinite subset of the tilings from {\em each} Wang tile set: changing the scale factor lets us simulate a new tiling, from a new, or the same, tile set.\footnote{As they point out, the {\em countable} restriction is somehow necessary,  because having a single tiling that simulates all tilings, is impossible, because there is a countable set of representation functions but an uncountable set of tilings.} Via our construction, obtain a single tile that is capable of these kinds of simulation.

\fi

\section*{Acknowledgements}
\ifabstract
\else
This work was initiated at the 27th Bellairs Winter Workshop
on Computational Geometry held on February 11-17, 2012 in Holetown, Barbados.
We thank the other participants of that workshop for a fruitful and
collaborative environment.  In particular, we thank Brad Ballinger and
Anna Lubiw for important discussions regarding Lemma~\ref{le:chain}. In addition, we thank Jarkko Kari for interesting and fruitful discussions on aperiodic tilings of the plane with few tile types.
\fi

\bibliographystyle{abbrv}
\ifabstract
\else
\bibliography{tam,tiling,jigsaw}
\fi

\end{document}